\documentclass[journal,twoside]{IEEEtran}

\usepackage{cite,url}
\usepackage{booktabs}
\usepackage{amsmath,amssymb,bm}
\usepackage{amsthm}
\usepackage{color}
\usepackage{graphicx}
\usepackage[lined,boxed,commentsnumbered, ruled, linesnumbered]{algorithm2e}
\usepackage{algpseudocode}

\usepackage{booktabs}
\usepackage{subfigure}
\usepackage{amsfonts}
\usepackage{threeparttable}

\newtheorem{lemma}{Lemma}

\newtheorem{corollary}{Corollary}

\begin{document}

\title{{Joint Rate Control and Power Allocation for Non-Orthogonal Multiple Access Systems}
\thanks{W. Bao is with the School of Information Technologies, The University of Sydney, Sydney, NSW 2006, Australia (email: wei.bao@sydney.edu.au). H. Chen, Y. Li, and B. Vucetic are with School of Electrical and Information Engineering, The University of Sydney, Sydney, NSW 2006, Australia (email: he.chen@sydney.edu.au, yonghui.li@sydney.edu.au, branka.vucetic@sydney.edu.au).}
}
\author{Wei Bao, He Chen,  Yonghui Li, and Branka Vucetic}

\maketitle

\begin{abstract}
This paper investigates the optimal resource allocation of a downlink non-orthogonal multiple access (NOMA) system consisting of one base station and multiple users. Unlike existing short-term NOMA designs that focused on the resource allocation for only the current transmission timeslot, we aim to maximize a long-term network utility by jointly optimizing the data rate control at the network layer and the power allocation among multiple users at the physical layer, subject to practical constraints on both the short-term and long-term power consumptions. To solve this problem, we leverage the recently-developed Lyapunov optimization framework to convert the original long-term optimization problem into a series of online rate control and power allocation problems in each timeslot. The power allocation problem, however, is shown to be non-convex in nature and thus cannot be solved with a standard method. However, we explore two structures of the optimal solution and  develop a dynamic programming based power allocation algorithm, which can derive a globally optimal solution, with a polynomial computational complexity. Extensive simulation results are provided to evaluate the performance of the proposed joint rate control and power allocation framework for NOMA systems, which demonstrate that the proposed NOMA design can significantly outperform multiple benchmark schemes, including orthogonal multiple access (OMA) schemes with optimal power allocation and NOMA schemes with non-optimal power allocation, in terms of average throughput and data delay.
\end{abstract}

%
\section{Introduction}
Due to the explosive traffic growth and the fast proliferation of Internet of Things (IoT), the fifth generation (5G) of cellular networks are expected to face unprecedented challenges including 1000-fold increase in system capacity, improved spectral efficiency, and massive connectivity with diverse service requirements~\cite{Andrews2014JSAC,Wunder2014CM}. In this context, non-orthogonal multiple access (NOMA), although not completely new to the wireless industry and research community \cite{Vanka12super}, has been regarded as a promising radio access technology for the 5G wireless communication systems~\cite{Dai2015CM,ding2015application,Islam2016survey, wei2016survey}, due to its unique capability of achieving a higher spectral efficiency and supporting a large number of concurrent transmissions over the same communication resource. In fact, multiuser superposition transmission (MUST), a two-user downlink scenario of NOMA, has been investigated for the third generation partnership project long-term evolution advanced (3GPP-LTEA) networks~\cite{Lee2016ICC}.

The current fourth generation (4G) of cellular communication systems and previous generations primarily adopted orthogonal multiple access (OMA) technologies, such as frequency-division multiple access (FDMA), time-division multiple access (TDMA), code-division multiple access (CDMA), and orthogonal frequency-division multiple access (OFDMA). In these OMA schemes, the resources are first split into orthogonal resource blocks in frequency/time/code domain and each resource block is then assigned to one user exclusively. As such, the inter-user interference can be avoided and the information of each user can be recovered at a low complexity. However, due to the orthogonal resource allocation, the maximum number of served users is limited by the total number of available resource blocks. In this sense, the OMA techniques are inadequate to support the massive connectivity requirement of 5G wireless systems. Another main problem of the OMA schemes is their relatively low spectral efficiency, especially for the resource blocks assigned to the users with poor instantaneous channel conditions. This issue can be effectively addressed by applying user selection schemes, where the users with strong channel conditions are selected out to transmit over the limited number resources. However, this may pose a serious user fairness problem as the users near to the cell edge could have much fewer transmission opportunities than those near to the cell center.

Unlike OMA with orthogonal resource allocation, NOMA advocates the usage of the power domain to multiplex signal streams from multiple users together and serve them simultaneously using the same frequency/time/code resource block. At the transmitter, NOMA adopts the superposition coding~\cite{Cover1972tit} to superimpose the signals of multiple users together by splitting them in the power domain. At the receiver side, successive interference cancellation (SIC)~\cite{Patel1994jsac} is implemented to separate multiplexed users' signals. In this case, each user can access to all resource blocks such that those resources that are solely assigned to users with poor channel quality in OMA can still be accessed by other users with good channel conditions in NOMA, which enables NOMA to achieve a higher spectral efficiency than OMA \cite{tse2005fundamentals}. Apart from this, NOMA is capable of realizing an improved tradeoff between system throughput and user fairness than OMA. This can be achieved by allocating less power to users with better channel conditions and more power to users with worse channel conditions, which is totally opposite to the conventional water-filling power allocation scheme widely used in OMA.

\subsection{Related Work}
The idea of using NOMA as a potential candidate for 5G multiple access technology was first proposed in~\cite{Saito2013vtc,Saito2013pimrc} by NTT DoCoMo, as a part of the mobile and wireless communications enablers for the twenty-twenty information society (METIS) projects. The system-level performance of NOMA was evaluated by simulation in \cite{Benjebbovu2013gcw} by taking into consideration many practical factors, including multiuser power allocation, signaling overhead, SIC error propagation and high user mobility. A test-bed of two-user downlink NOMA system was developed and evaluated in \cite{Benjebbour2015non}, where a NOMA scheme with each user occupying the whole bandwidth of 5.4 MHz was compared to an OMA scheme with each user using the transmission bandwidth of 2.7 MHz. Experimental results demonstrated that the NOMA scheme significantly outperforms its OMA counterpart in terms of both aggregate and individual user's throughput.

Since NOMA uses the power domain to realize multiple access and implements SIC to perform user decoding, the power allocation to the data flows of different users plays an important role in determining the performance of NOMA systems. Specifically, the total power of the transmitter should be allocated in a proper way so that the signals for the users with worse channel conditions can be successfully decoded and subsequently subtracted from the received signal of those users with better channel conditions. Early efforts on NOMA have mostly adopted the fixed power allocation scheme, in which the power allocation coefficients are pre-determined and are not affected by the instantaneous channel conditions~\cite{Ding2014SPL,Choi2014cl,Ding2015cl,Sun2015wcl,Ding2016twc}. The fixed power allocation strategy is favorable in terms of low implementation complexity, however, the system performance can be disadvantaged if the power allocation coefficients are not set appropriately. Inspired by the key idea of cognitive radio, another power allocation scheme was proposed and analyzed in~\cite{Ding2016tvt}, in which NOMA is treated as a special case of cognitive radio networks. In this scheme, the user with a weaker channel condition is regarded as a primary user, who has a higher priority to be served. That is, the available power is first allocated to fulfill the quality of service of the primary user, while the secondary user with a better channel condition will be opportunistically served using the remaining amount of power. In this case, the overall system spectral efficiency could be limited by the fact that the strong user's performance is highly affected by the weak user's channel quality. Recently, Yang \emph{et al.} proposed a general power allocation scheme to protect the QoS of both users, in which the power allocation factors are dynamically adjusted according to the instantaneous channel status to ensure that the rates of two users in NOMA are both larger than those in OMA~\cite{Yang2016cl,Yang2016twc}.

The aforementioned power allocation strategies were primarily designed for two-user NOMA scenarios. In general, more than two users can be multiplexed on the same resources, where the power allocation problem becomes more complicated. In \cite{Timotheou2015spl}, the dynamic power allocation scheme for a multi-user NOMA system was addressed with max-min user fairness for the cases with and without instantaneous channel state information (CSI). More specifically, the objective was set to maximize the minimum achievable user rate when the instantaneous CSI is available, and to minimize the maximum user outage probability when only average CSI is known. For the case without instantaneous CSI, a sum power minimization problem was resolved in \cite{Cui2016spl} subject to outage probabilistic constraints and the optimal decoding order. The work of \cite{Datta2016wcnc} provided an optimal power allocation solution to maximize the weighted sum rate of all users subject to a total power constraint.

When NOMA is applied in a multi-carrier system, the power allocation problem will be upgraded to a joint power and channel allocation problem. There have been some great efforts in this research line~\cite{Lei2015gc,Lei2016twc,Di2016twc,Zhang2016icc,Li2016cl,wei2016power,Sun2016tcom}. The joint power and channel allocation problem of multi-carrier NOMA systems was formally formulated with both maximum weighted sum rate and sum-rate utilities in \cite{Lei2015gc,Lei2016twc}, in which the problem was proved to be NP hard and a new algorithm inspired by Lagrangian duality and dynamic programming was proposed to achieve a near-optimal solution. In \cite{Di2016twc,Zhang2016icc}, the joint power and channel allocation problem was resolved by decoupling it as a many-to-many matching game with externalities and a geometric programming subproblem. Apart from the maximization of the sum rate or weighted sum rate, the total transmit power minimization problems were addressed subject to the predefined QoS of individual users in \cite{Li2016cl,wei2016power}. Very recently, Sun \emph{et al.} studied a full-duplex multi-carrier NOMA system, wherein the monotonic optimization was adopted to design an optimal joint power and subcarrier allocation policy such that the weighted sum system throughput is maximized.

\subsection{Motivation and Contributions}
We notice that most existing dynamic resource allocation strategies for NOMA systems have mainly focused on short-term (e.g., ``one-snapshot'') designs. Specifically, the system performance is normally optimized subject to a short-term peak power constraint for a single transmission block, overlooking the practical long-term power constraint. These short-term designs may lead to inferior system performance in a long-term perspective since the power resources are forced to be used even when the channel condition is not good enough in some transmission blocks. Furthermore, in practice, the amount of data that can be transmitted at the physical layer is highly influenced by the rate control at the network layer. The limited transmit power should not be wasted to those users who have little data to send/receive. In this sense, to achieve a higher long-term system throughput with limited resources, it is desirable to jointly design the rate control at the network layer and the resource allocation at the physical layer for NOMA systems. However, to our best knowledge, this important problem has not been considered in the open literature.

Motivated by this gap, in this paper we investigate the joint rate control and power allocation of a downlink NOMA system with one base station (BS) communicating to multiple users. Different from existing short-term resource allocation policies, we aim to maximize the long-term network utility, defined as a sum of concave functions of average data rates of all users, subject to a peak power constraint and a long-term average power constraint in the physical layer, as well as a peak rate constraint in the network layer. In the proposed design, we implement data queues at the BS to characterize users' incoming data flows from the Internet and outgoing data flows to the wireless channel. Note that it is desirable to jointly consider the data rate control and power optimization problems due to the fact that the amount of data delivered to each user in a certain timeslot is limited by not only the instantaneous channel capacity but also the amount of data available in the queue. Furthermore, the considered short-term peak power constraint and long-term average power constraint make the rate control and power allocation of different users tangled not only in each timeslot but also across various timeslots in the long term. In this sense, developing an online algorithm that can maximize the long-term network utility through optimizing the coupled data rate control and power allocation in each timeslot subject to both short-term and long-term power constraints is a challenging task.

The main contributions of this paper are summarized as follows:
\begin{itemize}
  \item We, for the first time, develop a systematic framework for the joint rate control and power allocation in NOMA systems. The formulated long-term optimization problem is fundamentally different from those short-term designs in the open literature since it is subject to \textit{long-term} average and \textit{short-term} peak power constraints, the network-layer peak rate constraints, and the queue stability constraints. As such, it cannot be resolved by the existing approaches. Motivated by this, we leverage the recently-developed Lyapunov optimization approach \cite{neely2006stochastic,neely2010stochastic}, which enables us to achieve the asymptotic optimality of the formulated problem. This is realized by converting the original long-term optimization problem to a series of online queue- and channel-aware optimization problems to be resolved at each single timeslot. Here, ``online'' means that the BS only needs to know the queue states and channel states at the current timeslot, without predicting anything in the future or knowing the statistics of the channel states.
  \item The online queue- and channel-aware optimization problem is then decomposed into rate control and power allocation problems in each timeslot. However, the power allocation problem is non-convex in nature and cannot be solved with a standard method. Instead, we explore two important structures in solving the problem.
      Referred to as the \textit{finite-point structure}, we show that the optimality only possibly happens at a finite set of candidate points,  which substantially decreases the searching space. Referred to as the \textit{incremental structure}, we derive the recursion of the original optimization problem and its subproblems, leading to the desired Bellman equation. As a consequence, we propose a new dynamic programming based power allocation (DPPA) approach to derive a globally optimal solution,  within a polynomial computational complexity.

  \item We evaluate the performance of our proposed framework for joint rate control and power allocation for NOMA systems via simulation and show that NOMA can greatly improve the network performance in both data rate and delay compared with four benchmark schemes, including OMA and non-optimal NOMA schemes, under a variety of system settings.

\end{itemize}

\subsection{Organization}
The rest of the paper is organized as follows. The system model and problem formulation are described in Section II, where the Lyapunov optimization approach is applied to convert the original long-term optimization problem into a series of rate control and power allocation problems to be optimized in each timeslot. In Section III, we elaborate the two special structures of the power allocation problems and develop a new dynamic programming based power allocation algorithm, which is guaranteed to achieve its global optimality. Simulation results and the associated discussions are presented in Section IV, and the conclusions of the paper are drawn in Section V.

\section{System Model and Lyapunov Optimization}\label{section_Lyapunov}

\begin{figure}[t]
\centering  \hspace{0pt}
\includegraphics[scale=0.4]{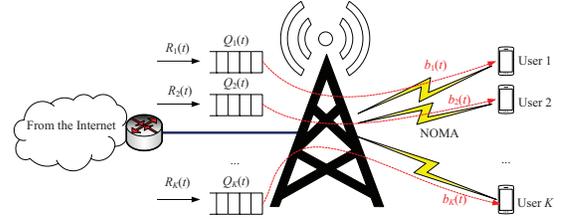}
\vspace{-0.2 cm}
\caption{System model of the considered NOMA system.}
\label{figure_model}
\end{figure}

\subsection{System Model}
As shown in Fig.~\ref{figure_model}, we consider the downlink data transmission from one BS to $K$ users using NOMA \cite{Ding2014SPL, Ding2016twc, Timotheou2015spl, Cui2016spl, Lei2016twc, Di2016twc}. All the nodes are equipped with single antenna and work in a half-duplex mode. The system operates on slotted time $t\in\{0,1,2,\ldots\}$.
At timeslot $t$, $R_i(t)$ bits of data for user $i$ arrive at the BS from the Internet. These data are firstly buffered at queue $i$ and then forwarded to user $i$ via a wireless channel. Let $Q_i(t)$ denote the amount of data buffered at queue $i$ at time $t$, i.e., queue backlog of the user $i$. We assume that the queue is large enough so that no data will be dropped. Let $b_i(t)$ denote the amount of data that can be delivered to the user $i$ at $t$ (i.e., data transmission capability offered by the underlying wireless channel).  $R_i(t)$ and $b_i(t)$ are variables to be designed by the system. $R_i(t)$ indicates how many data should be sent to the BS in the view of the network layer, and $b_i(t)$ indicates how many data should be sent from the BS to user $i$ via the wireless link in the view of the physical layer. Since the data rate of one user cannot be arbitrarily large in the network, we consider that $R_i(t)$ is limited to $R_{\max}$.\footnote{This value is imposed by an upper-layer protocol (e.g.,
TCP) for each user. This is because the upper layer protocol usually maintains a buffer to send
its data, and maximum possible data can be sent in each timeslot is limited by the buffer size. This assumption is widely
adopted in the literature that focuses on network flow control, such as \cite[Section 5.2]{neely2010stochastic} and
\cite[Section 5.2.4]{neely2006stochastic}.} $b_i(t)$ is also limited by finite transmission power, which will be elaborated shortly. Mathematically, we can write the following expression to characterize the evolution of each user's  queue backlog
%
%
\begin{align}
Q_i(t+1)=[Q_i(t)-b_i(t)]^{+}+R_i(t),
\end{align}
where $[\cdot]^+$ means $\max[\cdot, 0]$.

To improve the spectral efficiency, in this paper we consider that the BS adopts the NOMA method to transmit data to users in the downlink. As the first effort towards the long-term resource allocation designs for NOMA, here we concentrate on a single-carrier NOMA system such that the data transmissions to all users work on the same frequency band. The  proposed long-term resource allocation framework can be extended to a multi-carrier NOMA system, which we would like to consider as future work. In addition, it is assumed that channels from the BS to the users experience quasi-static and frequency flat fading such that the channel gains remain constant during each timeslot but change independently from one timeslot to another. We use $\widehat{g}_i(t)$ to denote the channel gain from the BS to user $i$ during timeslot $t$.  $\widehat{g}_i(t)$ are independent in different timeslots, but $\widehat{g}_i(t)$ and $\widehat{g}_j(t)$, $i\neq j$, can be dependent in one timeslot. As explained later, in our framework we do not need to know the probability distribution of $\widehat{g}_j(t)$. Let $p_i(t)$ denote the power allocated to transmit the data of user $i$ at $t$. Let $\widehat{\mathbf{g}}(t)\triangleq (\widehat{g}_1(t),\widehat{g}_2(t),\ldots, \widehat{g}_K(t))$ and $\mathbf{p}(t)\triangleq (p_1(t),p_2(t),\ldots, p_K(t))$. $\mathbf{p}(t)$ is a power allocation variable that determines how the power is allocated to multiple users at $t$. As we will see shortly, the values of $b_i(t)$ depend on channel gains and power allocations at $t$.

We define $g_{s(1)}(t), g_{s(2)}(t),\ldots, g_{s(K)}(t)$ as sorted $\widehat{g}_1(t),\ldots,\widehat{g}_K(t)$ in a descending order, where $s(i)$ indicates the index of the user with the $i$th largest channel gain, i.e., subscript of the user before sorting. According to the principle of NOMA with superposition coding and SIC, we can compute the amount of data that can be sent to user $s(i)$ (i.e., $b_{s(i)}(t)$) based on the channel gains and power allocation given below \cite{Timotheou2015spl,Lei2016twc}:
\begin{align}\label{formula_NOMA0}
\begin{cases}
b_{s(1)}(t)=&W\tau\log\left(1+\frac{p_{s(1)}(t)g_{s(1)}(t)}{\eta}\right),\\
b_{s(2)}(t)=&W\tau\log\left(1+\frac{p_{s(2)}(t)g_{s(2)}(t)}{p_{s(1)}(t)g_{s(2)}(t)+\eta}\right),\\
\ldots \\
b_{s(K)}(t)=&W\tau\log\left(1+\frac{p_{s(K)}(t)g_{s(K)}(t)}{\sum_{j=1}^{K-1}p_{s(j)}(t)g_{s(K)}(t)+\eta}\right),
\end{cases}
\end{align}
where  $W$ is the radio frequency bandwidth, $\tau$ is the duration of one timeslot, and $\eta$ is the noise level. $W$,  $\tau$, and $\eta$  are predetermined.


In reality, the BS transmission power can be constrained by two practical limitations. One could be imposed by the requirement for power savings, which limits the long-term average power consumption at the BS. The other may come from the regulations and rules (e.g., Federal Communication Commission (FCC)), which restricts the short-term transmission power at the BS. Motivated by this, we consider that the NOMA system is subject to a long-term average power constraint\footnote{The long-term average power constraint at the BS could be imposed by the mobile operator to control its electricity expense. It could also be imposed by the government agency to limit the energy consumption of BS to reduce its carbon emission on average. It is worth mentioning that this type of average power constraint has been commonly considered in existing literature that focused on the optimization of long-term performance of wireless systems, see e.g., \cite{Lau10TWC,Wei16TON}.},
\begin{align}\label{formula_queue_pconstraint1}
\lim_{T\rightarrow\infty}\frac{1}{T}\sum_{t=0}^{T-1}\mathbb{E}\left[\sum_{i=1}^{K}p_i(t)\right]\leq P_{\mathrm{mean}},
\end{align}
as well as a one-timeslot peak power constraint,
\begin{align}\label{formula_queue_pconstraint2}
\sum_{i=1}^{K}p_i(t)\leq P_{\max}, \forall t,
\end{align}
where $P_{\mathrm{mean}}$  and $P_{\max}$ are the maximum allowed long-term average power usage, and maximum allowed instantaneous power usage in each timeslot, respectively. Please note that the existing resource allocation designs for NOMA normally considered the short-term power constraint, but overlooked the long-term power constraint.

\subsection{Problem Statement}
To proceed, we first define the long-term network utility as follows
\begin{align}
\mathcal{U}=\sum_{i=1}^{K}U_i\left( \lim_{T\rightarrow\infty} \frac{1}{T}\sum_{t=0}^{T-1}\mathbb{E}[R_i(t)]\right),
\end{align}
where $U_i(\cdot)$ is the utility function of user $i$. It is an arbitrary concave non-decreasing function. User $i$'s utility is a function of its long-term average data rate, and the overall network utility is the sum of the individual user utilities.

The aim of this paper is to optimize the long-term network utility, by designing the data rate control in the network layer, $\mathbf{R}(t)$, and the power allocation in the physical layer, $\mathbf{p}(t)$. The problem is formally stated as the following Problem Origin (Problem PO):
\begin{subequations}
\begin{align}
\label{formula_PO_obj}\max_{\mathbf{R}(t), \mathbf{p}(t),\forall t} \ &\sum_{i=1}^{K}U_i\left( \lim_{T\rightarrow\infty}\frac{1}{T}\sum_{t=0}^{T-1}\mathbb{E}[R_i(t)]\right),   \\
\label{formula_PO_con1} \textrm{subject to }&\sum_{i=1}^{K}p_i(t)\leq P_{\max}, \forall t, \\
\label{formula_PO_con2}&p_i(t)\geq 0, \forall i,t, \\
\label{formula_PO_con3}&\lim_{T\rightarrow\infty}\frac{1}{T}\sum_{t=0}^{T-1}\mathbb{E}\left[\sum_{i=1}^{K}p_i(t)\right]\leq P_{\mathrm{mean}},\\
\label{formula_PO_con4}&0\leq R_i(t)\leq R_{\max}, \forall i,t, \\
\label{formula_PO_con5}& \lim_{t\rightarrow \infty}\frac{\mathbb{E}[Q_i(t)]}{t}=0, 
\end{align}
\end{subequations}
where $\mathbf{R}(t)\triangleq (R_1(t), R_2(t),\ldots, R_K(t))$. The last constraint $\lim_{t\rightarrow \infty}\frac{\mathbb{E}[Q_i(t)]}{t}=0$ means that $Q_i(t)$ is \textit{mean rate stable}. Its physical meaning is that the queue will not grow to infinity after sufficiently long period. Note that due to the considered queueing model given in (1) and the long-term average power constraint given in (3), the rate control and power allocation of all users across different timeslots are actually tangled together. As such, the optimal solution to the formulated optimization problem (6) cannot be achieved by simply maximizing the data rate in each timeslot.


\subsection{Lyapunov Optimization}
The Problem PO perfectly matches the Lyapunov optimization framework \cite{neely2006stochastic,neely2010stochastic}. This is because we aim to maximize a long-term utility under a long-term constraint as well as queue stability constraints. It can be shown that the long-term power constraint (\ref{formula_PO_con3}) can be converted to a virtual queue stable constraint \cite{neely2006stochastic, neely2010stochastic}. We define a virtual queue $Z(t)$, with $Z(0)=0$ and update equation
\begin{align}\label{formula_virtualqueue}
Z(t+1)=\left[Z(t)+\sum_{i=1}^{K}p_i(t)-P_{\mathrm{mean}}\right]^{+}.
\end{align}
With reference to \cite[Section 4.4]{neely2010stochastic}, (\ref{formula_PO_con3}) can be guaranteed if virtual queue $Z(t)$ is mean rate stable. As such, we replace condition (\ref{formula_PO_con3}) by ``$Z(t)$ is mean rate stable''. In fact, $Z(t)$ indicates the accumulated power debt at $t$: On average, the amount of $P_{\mathrm{mean}}$ power can be consumed in each timeslot. Using power more than $P_{\mathrm{mean}}$ generates power debt, and using power less than $P_{\mathrm{mean}}$ repays the power debt. The power debt must be maintained finite over infinite timeslots, so that the long-term power constraint is satisfied.

%

Subsequently, we study the queue evolution of the system. Let $\mathbf{Q}(t)\triangleq (Q_1(t), Q_2(t),\ldots, Q_K(t))$, and $\mathbf{\Theta}(t)=(\mathbf{Q}(t), Z(t))$. The Lyapunov function $\mathbf{L}(\mathbf{\Theta}(t))$ then becomes \cite{neely2006stochastic, neely2010stochastic}
\begin{align}
\mathbf{L}(\mathbf{\Theta}(t))\triangleq \frac{1}{2}\left[\sum_{i=1}^{K}Q_{i}^2(t)+Z^2(t)\right].
\end{align}
The Lyapunov drift is then defined as
\begin{align}
\Delta(\mathbf{\Theta}(t))\triangleq \mathbb{E}\left[\mathbf{L}(\mathbf{\Theta}(t+1))-\mathbf{L}(\mathbf{\Theta}(t))|\mathbf{\Theta}(t)\right].
\end{align}

Since we aim to maximize a long-term utility as well as to guarantee the queue stability, we focus on the following drift minus utility (drift plus penalty \cite{neely2006stochastic,neely2010stochastic})
\begin{align}
 \mathrm{DPP}(\mathbf{\Theta}(t))=\Delta(\mathbf{\Theta}(t))-V\mathbb{E}\left[\sum_{i=1}^{K}U_i(R_i(t))\Big|\mathbf{\Theta}(t)\right],
\end{align}
where $V$ is a tunable weight, representing the relative importance of ``utility maximization'' compared with ``queue stability''. Then $\mathrm{DPP}(\mathbf{\Theta}(t))$ can be further bounded in the following way \cite[Section 3.1.2]{neely2010stochastic}
\begin{footnotesize}
\begin{align}
\nonumber &\mathrm{DPP}(\mathbf{\Theta}(t))\leq
\mathbb{E}\left[\sum_{i=1}^{K}\frac{R_i^2(t)+b_i^2(t)}{2}\Big|\mathbf{\Theta}(t)\right]+
\mathbb{E}\left[\sum_{i=1}^{K}Q_i(t)R_i(t)\Big|\mathbf{\Theta}(t)\right]\\
\nonumber&-\mathbb{E}\left[\sum_{i=1}^{K}Q_i(t)b_i(t)\Big|\mathbf{\Theta}(t)\right]+\mathbb{E}\left[\frac{\left(\sum_{i=1}^{K}p_i(t)-P_{\mathrm{mean}}\right)^2}{2}\Big|\mathbf{\Theta}(t)\right]\\
\nonumber&+\mathbb{E}\left[Z(t)\left(\sum_{i=1}^{K}p_i(t)-P_{\mathrm{mean}}\right)\Big|\mathbf{\Theta}(t)\right]-V\mathbb{E}\left[\sum_{i=1}^{K}U_i(R_i(t))\Big|\mathbf{\Theta}(t)\right]\\
\nonumber&\leq B+
\mathbb{E}\left[\sum_{i=1}^{K}Q_i(t)R_i(t)\Big|\mathbf{\Theta}(t)\right]-\mathbb{E}\left[\sum_{i=1}^{K}Q_i(t)b_i(t)\Big|\mathbf{\Theta}(t)\right]
\\
&+\mathbb{E}\left[Z(t)\left(\sum_{i=1}^{K}p_i(t)-P_{\mathrm{mean}}\right)\Big|\mathbf{\Theta}(t)\right] -V\mathbb{E}\left[\sum_{i=1}^{K}U_i(R_i(t))\Big|\mathbf{\Theta}(t)\right],
\end{align}
\end{footnotesize}

\noindent where $B$ is an upper bound of $\mathbb{E}\left[\sum_{i=1}^{K}\frac{R_i^2(t)+b_i^2(t)}{2}\Big|\mathbf{\Theta}(t)\right]+\mathbb{E}\left[\frac{\left(\sum_{i=1}^{K}p_i(t)-P_{\mathrm{mean}}\right)^2}{2}\Big|\mathbf{\Theta}(t)\right]$. Please note that this value is bounded since $R_i(t)$, $b_i(t)$, and $p_i(t)$ are all bounded values.


Following the Lyapunov optimization approach, we need to ``opportunistically'' minimize the drift minus utility at each timeslot, denoted by $\mathrm{obj}(t)$,   i.e.,
the following term is minimized at each timeslot \cite{neely2006stochastic,neely2010stochastic}
\begin{align}
\nonumber&\mathrm{obj}(t)=\\
&\sum_{i=1}^{K}Q_i(t)R_i(t)-\sum_{i=1}^{K}Q_i(t)b_i(\mathbf{p}(t),\widehat{\mathbf{g}}(t))\\
\nonumber&+Z(t)\left(\sum_{i=1}^{K}p_i(t)-P_{\mathrm{mean}}\right)-V\sum_{i=1}^{K}U_i(R_i(t))\\
\label{formula_obj00}=&\sum_{i=1}^{K}\left[Q_i(t)R_i(t)-VU_i(R_i(t))\right]\\
\nonumber&-\left[\sum_{i=1}^{K}Q_i(t)b_i(\mathbf{p}(t),\widehat{\mathbf{g}}(t))-Z(t)\left(\sum_{i=1}^{K}p_i(t)-P_{\mathrm{mean}}\right)\right].
\end{align}
Please note that $b_i(t)$ is a function of $\mathbf{p}(t)$ and $\widehat{\mathbf{g}}(t)$ as shown in (\ref{formula_NOMA0}), so that we use $b_i(\mathbf{p}(t),\widehat{\mathbf{g}}(t))$ to represent $b_i(t)$ in the above equation. As a consequence, we arrive at the following single-timeslot optimization, which is referred to as Problem Single-Timeslot (Problem PST)
\begin{subequations}
\begin{align}
\min_{\mathbf{p}(t),\mathbf{R}(t)} \ &\mathrm{obj}(t),\\
\textrm{subject to } &\sum_{i=1}^{K} p_i(t)\leq P_{\max},\\
&p_i(t)\geq 0, \forall i,\\
\label{formula_pst_con3}&0\leq R_i(t)\leq R_{\max}, \forall i.
\end{align}
\end{subequations}
According to (\ref{formula_obj00})--(\ref{formula_pst_con3}), we notice that in order to solve the Problem PST,  the BS only needs to know the queue states $Q_i(t), \forall i$, virtual queue state $Z(t)$, and channel gains $\widehat{g}_i(t), \forall i$, in the current timeslot $t$. Therefore, the solution to Problem PST is an online queue- and channel-aware solution.

\subsection{Optimality of Lyapunov Optimization}\label{subsection_optimality_Lyaponov}
Our original aim is to solve the long-term optimization Problem PO. By employing the Lyapunov optimization approach, we convert the problem into the online single-timeslot optimization Problem PST. If the single-timeslot optimization is solved optimally at each timeslot, according to \cite[Section 5.4]{neely2006stochastic} and \cite[Section 5.1]{neely2010stochastic}, the optimal solution to the original Problem PO can then be achieved asymptotically, following an $[O(V), O(\frac{1}{V})]$ tradeoff between the queue backlogs and the achieved utility given below
\begin{align}
\label{formula_tradeoff1}\lim_{T\rightarrow \infty} \frac{1}{T}\sum_{t=0}^{T-1} \sum_{i=1}^{K}\mathbb{E}(Q_i(t))\leq& \frac{B}{\mu_{\mathrm{sym}}}+O(V),\\
\label{formula_tradeoff2}\sum_{i=1}^{K}U_i\left(\lim_{T\rightarrow \infty}\frac{1}{T}\sum_{t=0}^{T-1}\mathbb{E}[R_i(t)]\right)\geq& U^{*}-O\left(\frac{1}{V}\right),
\end{align}
where $U^{*}$ represents the maximum possible utility of Problem PO, $\mu_{\mathrm{sym}}$
 is a constant (predetermined by the system) representing the largest possible time average data rate that is simultaneously supportable to all users\footnote{$\mu_{\mathrm{sym}}$ represents the largest possible time average data rate that is simultaneously supportable to all user \cite[Definition 3.7 and Section 5.1.2]{neely2006stochastic}. Its physical meaning is explained as follows:
Suppose that the mean arrival rates of the $K$ users $(R_1(t), R_2(t),\ldots, R_K(t))$ are equal to $(\mu_1,\mu_2,\ldots,\mu_K)$. Then, $(\mu_1,\mu_2,\ldots,\mu_K)$ is in the network-layer capacity region if
these arrival rates can be stably supported by the network (i.e., all queues are stable), considering all possible strategies for choosing the control variables to affect scheduling and resource allocation. $(\mu_1,\mu_2,\ldots,\mu_K)$ is not in the network-layer capacity region if it is not possible to find any strategy to support these arrival rates. $\mu_{\mathrm{sym}}$ is the value such that $(\mu_{\mathrm{sym}}, \mu_{\mathrm{sym}},\ldots, \mu_{\mathrm{sym}})$ is in the capacity region but $(\mu_{\mathrm{sym}}+\epsilon, \mu_{\mathrm{sym}}+\epsilon,\ldots, \mu_{\mathrm{sym}}+\epsilon)$ is out of the capacity region for an arbitrary small $\epsilon$. For an arbitrary system,  $\mu_{\mathrm{sym}}$ exists and can be regarded as a predetermined constant.
 Please note that $\mu_{\mathrm{sym}}$ is usually a hidden value when we design the system. In this paper, we do not need to know $\mu_{\mathrm{sym}}$ to design rate control or power allocation. It is only used to show the $[O(V), O(\frac{1}{V})]$ tradeoff between the queue backlogs and the achieved utility.}, and $O(\cdot)$ represents the big O notation.

\subsection{Decomposition of Problem PST}

Since the objective function of Problem PST is in a sum form and the problem contains non-coupled constraints, Problem PST  can be decomposed into the following two subproblems. First, for each user $i$, we need to optimize
\begin{subequations}
\begin{align}
\max_{R_i(t)} \ &VU_i(R_i(t))-Q_i(t)R_i(t),\\
\textrm{subject to }&0\leq R_i(t)\leq R_{\max}.
\end{align}
\end{subequations}
The above problem is a rate control problem for user $i$. It is referred to as Problem RC-$i$ in the rest of this paper. Second, the following problem must also be solved
\begin{subequations}
\begin{align}
\max_{\mathbf{p}(t)} \ &\sum_{i=1}^{K}Q_i(t)b_i(\mathbf{p}(t),\widehat{\mathbf{g}}(t))-Z(t)\sum_{i=1}^{K}p_i(t),\\
\textrm{subject to } &\sum_{i=1}^{K} p_i(t)\leq P_{\max},\\
&p_i(t)\geq 0, \forall i.
\end{align}
\end{subequations}
This is a power allocation problem for all users. It is referred to as Problem PA in the rest of the paper. Problem PST is solved optimally if and only if Problems RC-$i$ and Problem PA are all solved optimally.

%

The solution to Problem RC-$i$ is straightforward. It is a single-variable optimization problem and the objective function is concave. The optimal $R_i(t)$ is simply equal to the solution to $V\frac{\partial{U_i}(t)}{\partial R_i(t)}=Q_i(t)$ (if it is grater than $R_{\max}$ or smaller than $0$, then the optimal $R_i(t)$ equals to $R_{\max}$ or $0$ respectively).

However, we can easily verify that Problem PA is a non-convex optimization problem, which cannot be solved through a standard method. For a general non-convex optimization problem, it is not guaranteed that there is an algorithm that can find a globally optimal solution within a polynomial computational complexity. Nevertheless,  in Section~\ref{section_DPPA}, we manage to develop two important structures for solving the problem, so that a globally optimal solution to Problem PA can be calculated within a polynomial computational complexity, through the proposed Dynamic Programming based Power Allocation (DPPA) algorithm.

\section{Dynamic Programming based Power Allocation}\label{section_DPPA}
In this section, we aim to solve the single-timeslot power allocation problem, i.e., Problem PA.
Our main contribution is to propose a new dynamic programming approach to derive its globally optimal solution. To this end, we first reformulate Problem PA for presentation convenience. Then, we explore two important properties of the optimization problem, which will be employed to construct the Bellman formula \cite[Chapter 15]{algorithmbook}. Finally, the Dynamic Programming based Power Allocation (DPPA) algorithm is proposed to find the optimal solution to Problem PA.

\subsection{Reformulation of Problem PA}
We first reformulate Problem PA to facilitate the presentation. First, problem PA is a single-timeslot optimization problem, which does not depend on other timeslots, so that we ignore the notation $t$ throughout this section. Second, without loss of generality, we assume that the channel gains follow $g_1\geq g_2\geq \ldots \geq g_K$. Let $\mathbf{g}\triangleq (g_1,g_2,\ldots, g_K)$. Please note that we do not make any specific assumptions on $Q_1, Q_2, \ldots, Q_K$, and $Z$ so that they could be arbitrary non-negative values. Also, without loss of generality, $W$ and $\tau$ are normalized to $1$ throughout this section. As a consequence, Problem PA is rewritten as
\begin{subequations}
\begin{align}
\nonumber\max_{\mathbf{p}} \ & Q_1 \log\left(1+\frac{p_1g_1}{\eta}\right)+Q_2 \log\left(1+\frac{p_2g_2}{p_1g_2+\eta}\right) +\ldots\\
\label{formula_PA_obj}&+ Q_K \log\left(1+\frac{p_Kg_K}{\sum_{j=1}^{K-1}p_jg_K+\eta}\right) -Z\left(\sum_{j=1}^{K}p_j\right),\\
\textrm{s.t. } & \sum_{j=1}^{K}p_j\leq P_{\max},\\
&p_k\geq 0, \forall k.
\end{align}
\end{subequations}
The physical meaning of the above Problem PA is interpreted as follows. First, the objective function (\ref{formula_PA_obj}) comprises a weighted sum of physical layer data rates of individual users. The weights are $Q_1, Q_2, \ldots, Q_K$ (i.e., the queue backlogs), which means that if more outstanding data are stored at queue $i$, then data transmission to user $i$ has a higher priority, and this is why its data rate is multiplied by the queue backlog $Q_i$. The objective function (\ref{formula_PA_obj}) also includes a penalty term $Z\cdot\sum_{j=1}^{K}p_j$, the product of $Z$ and total power consumption of the current timeslot. Recall from (\ref{formula_virtualqueue}) that $Z$ indicates the accumulated power debt. A larger $Z$ value will lead to a higher penalty in power consumption, so that less power will be consumed in the current timeslot, and thus the power debt could be kept finite, i.e., the long-term power constraint is satisfied. It is worth mentioning that the weighted sum rate maximization problem has been studied in \cite{Lei2015gc,Lei2016twc,Di2016twc}. However, their methods cannot address our power allocation problem due to the different objective function introduced by the considered long-term average power constraint.

In what follows, we introduce two important properties of Problem PA, which will enable to find a globally optimal solution.

\subsection{Finite-Point Structure}
In this subsection, we show the finite-point structure of Problem PA, which will greatly decrease the searching space for the optimal solution. Let $F$ denote the objective function (\ref{formula_PA_obj}).
\begin{lemma}
\textbf{Finite-Point Structure.} If the optimality of Problem PA is achieved, $\forall k\leq K$, the value of $p_{1}+p_{2}+\ldots+p_{k}$ can only be one from the following values in $\mathcal{E}$, where $\mathcal{E}=\bigg\{0; \frac{\eta (g_{j}Q_{j}-g_{i}Q_{i})}{g_{i}g_{j}(Q_{i}-Q_{j})}, \forall i, j, 1\leq i< j\leq K ; \frac{Q_{i}}{Z}-\frac{\eta}{g_{i}}, \forall i, 1\leq i\leq K ; P_{\max}\bigg\}$.
\end{lemma}

\begin{proof}

To proceed, we first calculate the first-order derivative of the objective function given by
\begin{small}
\begin{align}\label{formula_opt_diff}
\begin{cases}
\frac{\partial F}{\partial p_{1}}=&\frac{g_1Q_1}{p_1g_1+\eta}-\frac{g_2Q_2}{p_1g_2+\eta}+\frac{g_2Q_2}{(p_1+p_2)g_2+\eta}-\frac{g_3Q_3}{(p_1+p_2)g_3+\eta}+\\
&\ldots+\frac{g_{K-1}Q_{K-1}}{(p_1+\ldots+p_{K-1})g_{K-1}+\eta}-\frac{g_KQ_K}{(p_1+\ldots+p_{K-1})g_K+\eta}\\
&+\frac{g_KQ_K}{(p_1+\ldots+p_{K})g_K+\eta}-Z,\\
\frac{\partial F}{\partial p_{2}}=&\frac{g_2Q_2}{(p_1+p_2)g_2+\eta}-\frac{g_3Q_3}{(p_1+p_2)g_3+\eta}+\ldots\\
&+\frac{g_{K-1}Q_{K-1}}{(p_1+\ldots+p_{K-1})g_{K-1}+\eta}-\frac{g_KQ_K}{(p_1+\ldots+p_{K-1})g_K+\eta}\\
&+\frac{g_KQ_K}{(p_1+\ldots+p_{K})g_K+\eta}-Z,\\
 &\ldots\\
\frac{\partial F}{\partial p_{K-1}}=&\frac{g_{K-1}Q_{K-1}}{(p_1+\ldots+p_{K-1})g_{K-1}+\eta}-\frac{g_KQ_K}{(p_1+\ldots+p_{K-1})g_K+\eta}\\
&+\frac{g_KQ_K}{(p_1+\ldots+p_{K})g_K+\eta}-Z,\\
\frac{\partial F}{\partial p_{K}}=&\frac{g_KQ_K}{(p_1+\ldots+p_{K})g_K+\eta}-Z.
\end{cases}
\end{align}
\end{small}

If the optimality is achieved, the following Karush-Kuhn-Tucker (KKT) condition must be satisfied\footnote{Please note that the KKT condition does not guarantee global optimality, but global optimality implies the KKT condition. }.
\begin{align}\label{formula_KKT_condition}
\begin{cases}
-\frac{\partial F}{\partial p_k}-\lambda_k+\mu=0,\forall k,\\
\lambda_k=0 \textrm{ or }p_k=0, \forall k,\\
\mu=0 \textrm{ or } \sum_{j=1}^{K} p_j= P_{\max}.
\end{cases}
\end{align}
where $\lambda_k$ is the Lagrange multiplier associated with $p_k\geq 0$, and $\mu$ is the Lagrange multiplier associated with $\sum_{j=1}^{K} p_j\leq P_{\max} $.

At the optimal solution, if at least one entry of $\mathbf{p}$ is greater than zero\footnote{If all entries of the optimal $\mathbf{p}$ equal to zero, $\forall k\leq K$, the values of $p_{1}+p_{2}+\ldots+p_{k}$ are always $0$, so that the Lemma 1 is trivially true. This scenario can be ignored in this proof. }, we have $i_1<i_2<\ldots<i_M\in \{1,2,\ldots, K\}, M\geq 1$, $p_{i_1}>0, p_{i_2}>0, \ldots, p_{i_M}>0$. $\{i_1,i_2,\ldots, i_M\}$ is an arbitrary non-empty subset of $\{1,2,\ldots, K\}$. Then, following the KKT condition, we have $\lambda_{i_1}=\lambda_{i_2}=\ldots=\lambda_{i_M}=0$ and thus
\begin{align}\label{formula_opt_condition1}
\frac{\partial F}{\partial p_{i_1}}=\frac{\partial F}{\partial p_{i_2}}=\ldots=\frac{\partial F}{\partial p_{i_M}}.
\end{align}
If $\mu\neq0$, we have
\begin{align}\label{formula_opt_condition2_1}
p_{i_1}+p_{i_2}+\ldots+p_{i_{M-1}}=P_{\max}.
\end{align}
Otherwise, $\mu=0$, and we have
\begin{align}\label{formula_opt_condition2_2}
\frac{\partial F}{\partial p_{i_1}}=\frac{\partial F}{\partial p_{i_2}}=\ldots=\frac{\partial F}{\partial p_{i_M}}=0.
\end{align}
As a consequence, either (\ref{formula_opt_condition2_1}) or (\ref{formula_opt_condition2_2}) stands.

Substituting $p_j=0, j\neq i_1,\ldots,i_M $, into (\ref{formula_opt_diff}), and ignoring those lines regarding to $\frac{\partial F}{\partial p_j}, j\neq i_1,\ldots,i_M $, we have
\begin{small}
\begin{align}\label{formula_opt_diff2}
\begin{cases}
\frac{\partial F}{\partial p_{i_1}}=&\frac{g_{i_1}Q_{i_1}}{p_{i_1}g_{i_1}+\eta}-\frac{g_{i_2}Q_{i_2}}{p_{i_1}g_{i_2}+\eta}\\
&+\frac{g_{i_2}Q_{i_2}}{(p_{i_1}+p_{i_2})g_{i_2}
+\eta}-\frac{g_{i_3}Q_{i_3}}{(p_{i_1}+p_{i_2})g_{i_3}+\eta}+\ldots\\
&+\frac{g_{{i_{M-1}}}Q_{i_{M-1}}}{(p_{i_1}+\ldots+p_{i_{M-1}})g_{i_{M-1}}+\eta}-\frac{g_{i_{M}}Q_{i_{M}}}{(p_{i_1}+\ldots+p_{i_{M-1}})g_{i_{M}}+
\eta}\\
&+\frac{g_{i_{M}}Q_{i_{M}}}{(p_{i_1}+\ldots+p_{i_{M}})g_{i_{M}}+\eta}-Z,\\
\frac{\partial F}{\partial p_{i_2}}=&\frac{g_{i_2}Q_{i_2}}{(p_{i_1}+p_{i_2})g_{i_2}
+\eta}-\frac{g_{i_3}Q_{i_3}}{(p_{i_1}+p_{i_2})g_{i_3}+\eta}+\ldots\\
&+\frac{g_{{i_{M-1}}}Q_{i_{M-1}}}{(p_{i_1}+\ldots+p_{i_{M-1}})g_{i_{M-1}}+\eta}-\frac{g_{i_{M}}Q_{i_{M}}}{(p_{i_1}+\ldots+p_{i_{M-1}})g_{i_{M}}+
\eta}\\
&+\frac{g_{i_{M}}Q_{i_{M}}}{(p_{i_1}+\ldots+p_{i_{M}})g_{i_{M}}+\eta}-Z,\\
&\ldots\\
\frac{\partial F}{\partial p_{i_{M-1}}}=&\frac{g_{{i_{M-1}}}Q_{i_{M-1}}}{(p_{i_1}+\ldots+p_{i_{M-1}})g_{i_{M-1}}+\eta}-\frac{g_{i_{M}}Q_{i_{M}}}{(p_{i_1}+\ldots+p_{i_{M-1}})g_{i_{M}}+
\eta}\\
&+\frac{g_{i_{M}}Q_{i_{M}}}{(p_{i_1}+\ldots+p_{i_{M}})g_{i_{M}}+\eta}-Z,\\
\frac{\partial F}{\partial p_{i_{M}}}=&\frac{g_{i_{M}}Q_{i_{M}}}{(p_{i_1}+\ldots+p_{i_{M}})g_{i_{M}}+\eta}-Z.
\end{cases}
\end{align}
\end{small}

By combining (\ref{formula_opt_condition1}) into (\ref{formula_opt_diff2}), we have
\begin{align}
\begin{cases}
\frac{g_{i_1}Q_{i_1}}{p_{i_1}g_{i_1}+\eta}-\frac{g_{i_2}Q_{i_2}}{p_{i_1}g_{i_2}+\eta}=0,\\
\frac{g_{i_2}Q_{i_2}}{(p_{i_1}+p_{i_2})g_{i_2}+\eta}-\frac{g_{i_3}Q_{i_3}}{(p_{i_1}+p_{i_2})g_{i_3}+\eta}=0,\\
\ldots\\
\frac{g_{{i_{M-1}}}Q_{i_{M-1}}}{(p_{i_1}+\ldots+p_{i_{M-1}})g_{i_{M-1}}+\eta}-\frac{g_{i_{M}}Q_{i_{M}}}{(p_{i_1}+\ldots+p_{i_{M-1}})g_{i_{M}}+
\eta}=0,
\end{cases}
\end{align}
which leads to
\begin{align}\label{formula_opt_possibility1}
\begin{cases}
p_{i_1}=\frac{\eta (g_{i_2}Q_{i_2}-g_{i_1}Q_{i_1})}{g_{i_1}g_{i_2}(Q_{i_1}-Q_{i_2})},\\
p_{i_1}+p_{i_2}=\frac{\eta (g_{i_3}Q_{i_3}-g_{i_2}Q_{i_2})}{g_{i_2}g_{i_3}(Q_{i_2}-Q_{i_3})},\\
\ldots\\
p_{i_1}+p_{i_2}+\ldots+p_{i_{M-1}}=\frac{\eta (g_{i_M}Q_{i_M}-g_{i_{M-1}}Q_{i_{M-1}})}{g_{i_{M-1}}g_{i_M}(Q_{i_{M-1}}-Q_{i_M})}.\\
\end{cases}
\end{align}

By combining (\ref{formula_opt_condition2_1}) and (\ref{formula_opt_condition2_2}) into the last line of (\ref{formula_opt_diff2}), we have either
\begin{align}\label{formula_opt_possibility2}
\frac{\partial F}{\partial p_{i_M}}&=\frac{g_{i_{M}}Q_{i_{M}}}{(p_{i_1}+\ldots+p_{i_{M}})g_{i_{M}}+\eta}-Z=0,
\end{align}
or
\begin{align}\label{formula_opt_possibility3}
p_{i_1}+\ldots+p_{i_{M}}=P_{\max}.
\end{align}

Therefore, by combining (\ref{formula_opt_possibility1}), (\ref{formula_opt_possibility2}), and (\ref{formula_opt_possibility3}),
$p_{i_1}+p_{i_2}+\ldots+p_{i_m}$, $m\leq M$ can only possibly be one from the following values in $\mathcal{E}(i_1,i_2,\ldots, i_M) =\bigg\{\frac{\eta (g_{i_2}Q_{i_2}-g_{i_1}Q_{i_1})}{g_{i_1}g_{i_2}(Q_{i_1}-Q_{i_2})},
\frac{\eta (g_{i_3}Q_{i_3}-g_{i_2}Q_{i_2})}{g_{i_2}g_{i_3}(Q_{i_2}-Q_{i_3})},\ldots,\\
\frac{\eta (g_{i_M}Q_{i_M}-g_{i_{M-1}}Q_{i_{M-1}})}{g_{i_{M-1}}g_{i_M}(Q_{i_{M-1}}-Q_{i_M})}, \frac{Q_{i_M}}{Z}-\frac{\eta}{g_{i_M}}, P_{\max}\bigg\}$.

 Finally, $\{i_1,i_2,\ldots, i_M\}$ could be an arbitrary subset of $\{1,2,\ldots, K\}$. Thus, $p_{1}+p_{2}+\ldots+p_{k}$ can only possibly be $0$ or one value in  $\bigcup_{\textrm{all possible } i_1,i_2,\ldots, i_M} \mathcal{E}(i_1,i_2,\ldots, i_M)$, which is equivalent to  $\bigg\{0; \frac{\eta (g_{j}Q_{j}-g_{i}Q_{i})}{g_{i}g_{j}(Q_{i}-Q_{j})}, \forall i, j, 1\leq i< j \leq K; \frac{Q_{i}}{Z}-\frac{\eta}{g_{i}}, \forall i, 1\leq i \leq K$; $P_{\max}\bigg\}$. This completes the proof.

%
%
%
%
%

\end{proof}

Lemma 1 shows that when the optimality of Problem PA is achieved, the value of $p_{1}+p_{2}+\ldots+p_{k}, \forall k\leq K$, must be one from the candidate values in $\mathcal{E}$. This is because if a value of $p_{1}+p_{2}+\ldots+p_{k}, \forall k\leq K$, is not in $\mathcal{E}$, the KKT condition is not satisfied and $\mathbf{p}$ is not an optimal solution. Such property greatly decreases the searching space for the optimal solution, which will then be employed as a key in the DPPA algorithm.

Please note that if a value in $\mathcal{E}$ is greater than $P_{\max}$ or smaller than $0$,  $p_{1}+p_{2}+\ldots+p_{k}, \forall k\leq K$, cannot be equal to that value. Therefore, we can eliminate such values from $\mathcal{E}$. Let  $\widetilde{\mathcal{E}}\triangleq \{e\in\mathcal{E}| 0\leq e\leq P_{\max} \} $. We have the following corollary

\begin{corollary}\label{corollary}
If the optimality of Problem PA is achieved, the value of $p_{1}+p_{2}+\ldots+p_{k}, \forall k\leq K$, can only be one from those in $\widetilde{\mathcal{E}}$.
\end{corollary}
In the rest of this paper, let $L=|\widetilde{\mathcal{E}}|$ represent the number of elements in $\widetilde{\mathcal{E}}$. $L\leq \frac{K(K-1)}{2}+K+2$ since there are $\frac{K(K-1)}{2}+K+2$ elements in $\mathcal{E}$ by definition.

\subsection{Incremental Structure}

In this subsection, we present the incremental structure in the optimal solution, which will help to construct the Bellman equation to be used in the dynamic programming.

For convenience, we define
\begin{align}
&f_1(p_1)\triangleq Q_1 \log\left(1+\frac{p_1g_1}{\eta}\right)-Z\cdot p_1,\\
&f_2(p_1, p_2)\triangleq Q_2\log\left(1+\frac{p_2g_2}{p_1g_2+\eta}\right)-Z\cdot p_2,\\
\nonumber&\ldots\\
\nonumber &f_K(p_1+p_2+\ldots+p_{K-1}, p_K)\triangleq \\
&Q_K \log\left(1+\frac{p_Kg_K}{\sum_{j=1}^{K-1}p_jg_K+\eta}\right)-Z\cdot p_K.
\end{align}

Let
\begin{align}
\nonumber G_{k}(p_1,p_2,\ldots, p_k)\triangleq &f_1(p_1)+f_2(p_1, p_2)+\ldots \\
+&f_k(p_1+p_2+\ldots+p_{k-1}, p_k), \forall k.
\end{align}

Then, we have
\begin{align}
\nonumber G_{k}(p_1,p_2,\ldots, p_k)= &G_{k-1}(p_1,p_2,\ldots, p_{k-1})\\
+&f_k(p_1+p_2+\ldots+p_{k-1}, p_k).
\end{align}
Please note that $G_{K}(p_1,p_2,\ldots, p_K)=F(p_1,p_2,\ldots, p_K)$, the objective function of Problem PA.

We notice that if we want to obtain the optimal solution to PA, we can solve them by gradually improving the objective function, using the following incremental structure.

\begin{lemma}\textbf{Incremental Structure.}
If $p_1^*,p_2^*,\ldots, p_k^*$, $k\geq 2$, is an optimal solution to the following problem
\begin{subequations}
\begin{align}
\label{formula_lemma2_obj1_1}\max_{p_1,p_2,\ldots,p_k} &G_{k}(p_1,p_2,\ldots, p_k), \\
\label{formula_lemma2_obj1_2}\textrm{subject to } &p_1+p_2+p_3+\ldots+p_k= A,\\
\label{formula_lemma2_obj1_3}&p_j\geq0, \forall j=1,2,\ldots k,
\end{align}
\end{subequations}
where $A$ is some non-negative value, then, $p_1^*,p_2^*,\ldots, p_{k-1}^*$ is an optimal solution to the following problem
\begin{subequations}
\begin{align}
\label{formula_lemma2_obj2_1}\max_{p_1,p_2,\ldots,p_{k-1}} &G_{k-1}(p_1,p_2,\ldots, p_{k-1}),\\
\label{formula_lemma2_obj2_2}\textrm{subject to } &p_1+p_2+p_3+\ldots+p_{k-1}\\
\nonumber&= p_1^*+p_2^*+p_3^*+\ldots+p_{k-1}^*,\\
\label{formula_lemma2_obj2_3}&p_j\geq0, \forall j=1,2,\ldots k-1.
\end{align}
\end{subequations}
\end{lemma}

\begin{proof}
Suppose $p_1^*,p_2^*,\ldots, p_{k-1}^*$ is not an optimal solution to problem (\ref{formula_lemma2_obj2_1})--(\ref{formula_lemma2_obj2_3}), then we can find another solution $p_1^{**},p_2^{**},\ldots, p_{k-1}^{**}$ leading to a even larger value of the object function  (\ref{formula_lemma2_obj2_1}), $f_1(p_1^{**})+f_2(p_1^{**},p_2^{**})+\ldots+f_k(p_1^{**}+p_2^{**}+\ldots+p_{k-2}^{**},p_{k-1}^{**})>
f_1(p_1^{*})+f_2(p_1^{*},p_2^{*})+\ldots+f_k(p_1^{*}+p_2^{*}+\ldots+p_{k-2}^{*},p_{k-1}^{*})$.

Please note that $p_1^{**},p_2^{**},\ldots, p_{k-1}^{**}, p_k^{*}$ is also a feasible solution to (\ref{formula_lemma2_obj1_1})--(\ref{formula_lemma2_obj1_3}), however, such solution leads to a greater value in the objective function in (\ref{formula_lemma2_obj1_1}), i.e.,  $f_1(p_1^{**})+f_2(p_1^{**},p_2^{**})+f_2(p_1^{**}+p_2^{**},p_3^{**})+\ldots+f_k(p_1^{**}+p_2^{**}+\ldots+p_{k-1}^{**},p_k^{*})>
f_1(p_1^{*})+f_2(p_1^{*},p_2^{*})+f_2(p_1^{*}+p_2^{*},p_3^{*})+\ldots+f_k(p_1^{*}+p_2^{*}+\ldots+p_{k-1}^{*},p_k^{*})$, leading to a contradiction. Through this contradiction, Lemma 2 is proved.
\end{proof}

Lemma 2 implies that in order to solve the  optimization problem with the objective function $G_k(\cdot)$, we can first solve a reduced subproblem with the objective function $G_{k-1}(\cdot)$. Therefore, the original problem with objective function $G_K(\cdot)$ can then be solved recursively, through solving optimization problems with objective functions $G_{K-1}(\cdot), G_{K-2}(\cdot), \ldots, G_{1}(\cdot)$.
In what follows, we realize such idea by detailed derivations: We first derive the Bellman equation showing the recursion  of the optimality. Then, a dynamic programming approach is proposed to solve the Problem PA.


\subsection{Recursion of Optimality: Bellman Equation}

In the next step, we aim to solve the Problem PA via a dynamic programming approach. Intrinsically, dynamic programming is a mathematical optimization method. The original problem is broken down into simpler sub-problems in a recursive manner, so that the original problem can be solved optimally through recursively finding the optimal solutions to the sub-problems. The mathematical relation between the optimality of the problem and its sub-problems is characterized by the Bellman equation \cite[Chapter 15]{algorithmbook}.

As a prerequisite for dynamic programming, the Bellman equation  regarding to the optimal solution to Problem PA is derived in this subsection, which is achieved through taking the advantage of the structures derived by Lemmas 1 and 2.


We define the Subproblem-$l$-$k$ (SP-$l$-$k$) as follows
\begin{subequations}
\begin{align}
\max_{p_1,p_2,\ldots,p_k} &G_k(p_1,p_2,\ldots, p_k),\\
\textrm{subject to } &p_1+p_2+p_3+\ldots+p_{k}= \pi_l,\\
&p_j\geq0, \forall j=1,2,\ldots  k,\\
&p_{1}+p_{2}+\ldots+p_{k'}\in\widetilde{\mathcal{E}},\forall k'\leq k,
\end{align}
\end{subequations}
where $\pi_l$ is the $l$th smallest element in $\widetilde{\mathcal{E}}$. $\{\pi_1,\pi_2,\ldots,\pi_L\}$ is the sorted version (in ascending order) of $\widetilde{\mathcal{E}}$.
We define $H(l,k)$ as the  optimal $G_k$ value of SP-$l$-$k$, and let $\mathbf{p}_{opt}(l,k)=(p_{opt,1}(l,k),\ldots, p_{opt,k}(l,k))$ denote the optimal solution. Following Corollary 1, since $p_1+p_2+p_3+\ldots+p_{K}$ is one from the values in $\{\pi_1,\pi_2,\ldots,\pi_L\}$,  the optimal solution to Problem PA is the best one among the optimal solutions to SP-$1$-$K$, SP-$2$-$K$, \ldots, SP-$L$-$K$.

Then, we characterize  $H(l,k)$ through analyzing $H(l',k-1)$.
When SP-$l$-$k$ is optimized,  $p_1+p_2+p_3+\ldots+p_{k-1}$ must be equal to some $\pi_{l'}\in\widetilde{\mathcal{E}}$, $l'\leq l$. Given $l'$, due to Lemma 2, the following SP-$l'$-$(k-1)$ must be optimized first
\begin{subequations}
\begin{align}
\max_{p_1,p_2,\ldots,p_{k-1}} &G_{k-1}(p_1,p_2,\ldots, p_{k-1}),\\
\textrm{subject to } &p_1+p_2+p_3+\ldots+p_{k-1}= \pi_{l'},\\
&p_j\geq0, \forall j=1,2,\ldots  k-1,\\
&p_{1}+p_{2}+\ldots+p_{k'}\in\widetilde{\mathcal{E}},\forall k'\leq k-1.
\end{align}
\end{subequations}
However, at this stage, we do not know the value $l'$, so that we need to examine all possible $l'$ values. By solving all possible SP-$l'$-$(k-1)$, $\forall l'\leq l$, we can find the best one among them to achieve the optimal SP-$l$-$k$.

Given that SP-$l'$-$(k-1)$ is optimized (it implies $p_1+p_2+p_3+\ldots+p_{k-1}= \pi_{l'}$), and given the constraint $p_1+p_2+p_3+\ldots+p_{k}= \pi_{l}$ in SP-$l$-$k$,  $p_k$ is fixed at $\pi_l-\pi_{l'}$. Then the value of the objective function of SP-$l$-$k$ becomes $H(l',k-1)+f_k(\pi_{l'}, \pi_l-\pi_{l'})$ in this scenario. By examining all possible $l'$, $\forall l'\leq l$, $H(l,k)$ is the largest one among all possible $H(l',k-1)+f_k(\pi_{l'}, \pi_l-\pi_{l'})$ values, $l'\leq l$.

As a consequence,  we have the following Bellman equation showing the relationship between the optimal solution to SP-$l$-$k$ and its subproblems for $k\geq 2$
\begin{align}\label{formula_bellman}
H(l,k)=\max_{l'=1,2,\ldots, l}(H(l',k-1)+f_k(\pi_{l'},\pi_l-\pi_{l'})).
\end{align}

When $k=1$, the optimal value of the objective of SP-$l$-$1$ can be simply derived as $H(l,1)=f_1(\pi_l)$ following definition.

The optimal solution to Problem PA is $\max \left[H(1,K), H(2,K), \ldots, H(L,K)\right]$, which can be derived through recursively finding the optimal solutions to $H(l,K-1), H(l,K-2), \ldots, H(l,1), \forall l$, by employing dynamic programming. In the next subsection, we will formally propose the Dynamic Programming based Power Allocation (DPPA) algorithm to solve Problem PA. 


\begin{algorithm}[h]\label{algorim1}
\CommentSty{Find all possible $p_1+p_2+\ldots+p_k, k\leq K$, values and put them in $\widetilde{\mathcal{E}} $}\;\label{algoritmline_label1}
$\widetilde{\mathcal{E}} \leftarrow \{0, P_{\max}\}$\;
\For{$i\leftarrow 1$ \KwTo $K$}
{
    \For{$j\leftarrow i+1$ \KwTo $K$}
    {
        $p_a\leftarrow \frac{\eta (g_{j}Q_{j}-g_{i}Q_{i})}{g_{i}g_{j}(Q_{i}-Q_{j})} $\;
        \If{$0< p_a< P_{\max}$}
        {
            Put $p_a$ in $\mathcal{E}$\;
        }
    }

    $p_a\leftarrow \frac{Q_{i}}{Z}-\frac{\eta}{g_{i}}$\;
    \If{$0< p_a< P_{\max}$}
    {
        Put $p_a$ in $\widetilde{\mathcal{E}} $\;
    }

}
$L\leftarrow |\widetilde{\mathcal{E}} |$, i.e., the size of $\widetilde{\mathcal{E}} $\;
Sort all elements in $\widetilde{\mathcal{E}}$ in ascending order $\pi_1, \pi_2,\ldots, \pi_L$\; \label{algoritmline_label2}
\CommentSty{Dynamic programming using Bellman equation}\;  \label{algoritmline_label3}
Set $H(l,k)\leftarrow 0$, $\forall k=1,2,\ldots, K, \forall l=1,2,\ldots, L$\;
Set $H(l,1)\leftarrow f_1(\pi_l)$, $\mathbf{p}_{opt}(l,1)\leftarrow (\pi_l)$, $\forall l=1,2,\ldots, L$\;
\For{$k\leftarrow 2$ \KwTo $K$\label{algoritmline_loop1}}
{
    \For{$l\leftarrow 1$ \KwTo $L$ \label{algoritmline_loop2}}
    {
        $H(l,k)\leftarrow \max_{l'=1,2,\ldots,l}\left[H(l',k-1)+f_k(\pi_l',\pi_l-\pi_l')\right]$\; \label{algoritmline_loop3}
        $l_0\leftarrow \arg \max_{l'=1,2,\ldots,l}\left[H(l',k-1)+f_k(\pi_l',\pi_l-\pi_l')\right]$\; \label{algoritmline_loop4}
        $\mathbf{p}_{opt}(l,k)\leftarrow (\mathbf{p}_{opt}(l_0,k-1),\pi_l-\pi_{l_0})$;  \CommentSty{append $\pi_l-\pi_{l_0}$ in the end}\;
    }
}\label{algoritmline_label4}
$H_{opt}\leftarrow \max_{l=1,2,\ldots,L}H(l,K)$; \CommentSty{Optimal value of objective function of Problem PA}\;   \label{algoritmline_label5}
$l_0\leftarrow \arg \max_{l=1,2,\ldots,L}H(l,K)$\;
$\mathbf{p}_{opt}\leftarrow \mathbf{p}_{opt}(l_0,K)$; \CommentSty{Optimal power allocation of Problem PA.}   \label{algoritmline_label6}
\caption{Dynamic Programming based Power Allocation}
\end{algorithm}

\subsection{Dynamic Programming based Power Allocation (DPPA)}

In this subsection, we propose the DPPA formally described in Algorithm~\ref{algorim1} through the derived the Bellman equation given in (\ref{formula_bellman}).


%
%
%
%
%

%
%

%
%
%
%
%
%
%
%
%
%
%
In Lines~\ref{algoritmline_label1}--\ref{algoritmline_label2} of Algorithm 1, we construct the set $\widetilde{\mathcal{E}}$. In Lines~\ref{algoritmline_label3}--\ref{algoritmline_label4}, we derive the optimal $H(k,l)$ values via dynamic programming approach. Finally, the optimal solution to Problem PA is derived in Lines~\ref{algoritmline_label5}--\ref{algoritmline_label6}, since the optimal solution to Problem PA is the best one among the optimal solutions to SP-$1$-$K$, SP-$2$-$K$, \ldots, SP-$L$-$K$.

\subsection{Complexity Analysis}

The computational complexity of the above DPPA is characterized as follows. In Lines~\ref{algoritmline_loop1}--\ref{algoritmline_label4}, there are two levels of ``for'' loops. The outer level (Line~\ref{algoritmline_loop1}) introduces a computational complexity of $O(K)$, and the inner level (Line~\ref{algoritmline_loop2}) introduces a computational complexity of $O(L)$. Within the inner ``for'' loop (Lines \ref{algoritmline_loop3}--\ref{algoritmline_loop4}), we search for the maximum term among $l$ terms, and $l\leq L$, introducing another level of computational complexity of $O(L)$. Therefore, the overall computational complexity is $O(KL^2)$. $L$ is in the scale of $O(K^2)$ according to Corollary~\ref{corollary}, so that the overall computational complexity is $O(K^5)$, which is polynomial.


Note that for a general non-convex optimization problem, there are no standard methods to find a globally optimal solution within a polynomial computational complexity. In order to derive a globally optimal solution to Problem PA, one alternative method is to study the KKT condition of the optimization problem and exhaustively search for  all the solutions satisfying the KKT condition. Then, the best solution among them is selected as the globally optimal solution.

 By this method, the KKT condition is shown in (\ref{formula_KKT_condition}), and then we need to find all solutions satisfying (\ref{formula_KKT_condition}).
According to the second line of (\ref{formula_KKT_condition}), for each user $k$, we need to consider two possibilities: (1) $\lambda_k=0$ and (2) $p_k=0$ for this user. Since there are $K$ users, there are $2^K$ possibilities. In the third line of (\ref{formula_KKT_condition}), we need to consider another set of two possibilities. In sum, we need to exhaustively search for  $2\times2^{K}$ possibilities to find all solutions satisfying (\ref{formula_KKT_condition}), in order to find the best one among them. Therefore, the computational complexity is $O(2^{K})$, non-polynomial, without our proposed DPPA algorithm.


%

\section{Simulation Performance Evaluation}

In this section, we present simulation based performance evaluation of the proposed long-term resource allocation framework for NOMA. We first introduce the simulation setup and four benchmark schemes, and then demonstrate the performance gain of our DPPA-based NOMA scheme over these benchmark schemes. For simplification, the NOMA scheme proposed in this paper is referred to as NOMA-OPT in the rest of this section.

\begin{figure}[t]
\centering  \hspace{0pt}
\includegraphics[scale=0.45]{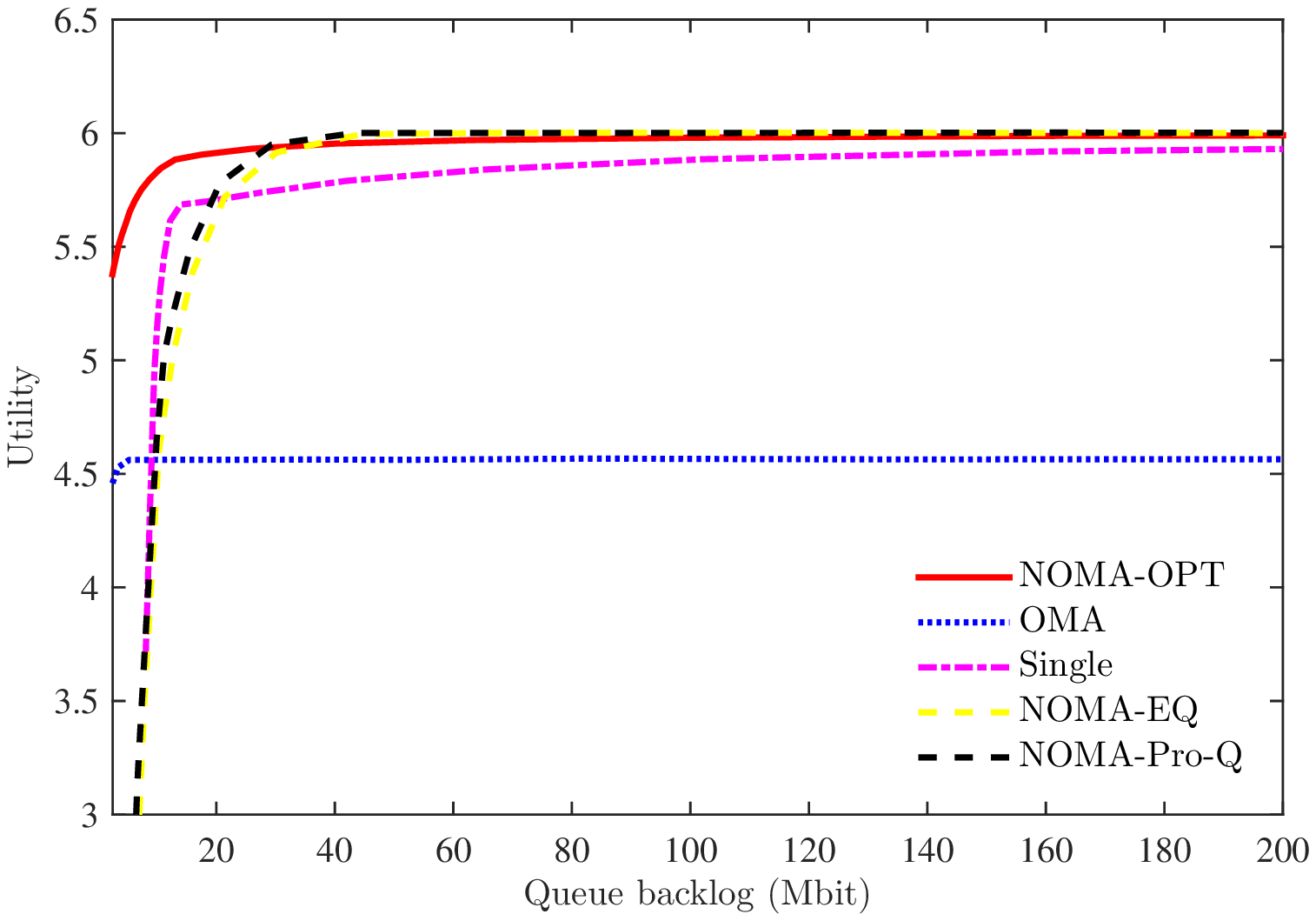}
\vspace{-0.2 cm}
\caption{Queue backlog versus network utility in Scenario 1.}
\label{figure_sim1}
\centering  \hspace{0pt}
\includegraphics[scale=0.45]{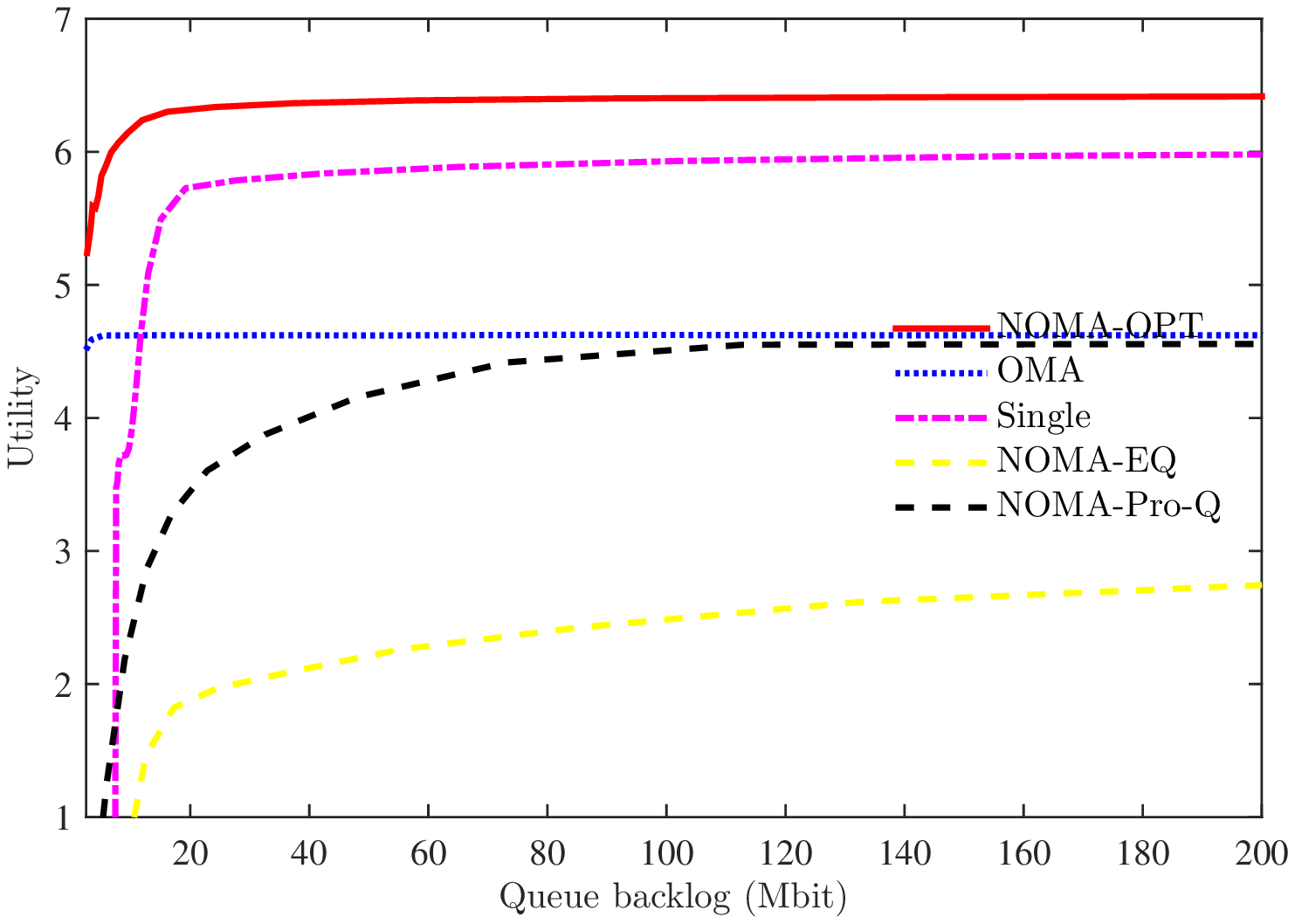}
\vspace{-0.2 cm}
\caption{Queue backlog versus network utility in Scenario 2.}
\label{figure_sim2}
\end{figure}

\subsection{Simulation Setup}
The channel gain from the BS to each user is computed as $\frac{1}{D^{\alpha}}\cdot g_0$, where $D^{\alpha}$ is the pathloss, $D$ (in meter) is the distance between the BS to the user, $\alpha=4$ is the pathloss exponent, and $g_0$ is independently exponentially distributed with a unit mean (corresponding to a normalized Rayleigh fading). We set $W=20$ MHz and $\tau=50$ ms. The utility function of each user is set to be $\ln(\textrm{data rate in Mbit per timeslot})$.  Each simulation point is averaged over $50000$ timeslots. We set $P_{\max}=33$ dBm, $P_{\mathrm{mean}}=30$ dBm, $R_{\max}=15$ Mbits, and the noise $\eta=-87$ dBm.  The initial queue backlogs for all users are $0$.

\subsection{Benchmark Schemes}
In this section, we compare the proposed NOMA-OPT scheme with four benchmark schemes as follows.
\subsubsection{Time-sharing Orthogonal Multiple Access (OMA) Scheme} In this scheme,  different users are operated on equally partitioned time durations in each timeslot. For fairness,  the OMA scheme has the same queueing model as the NOMA-OPT scheme: Problem PO is still the initial objective to be solved. The same Lyapunov optimization analysis can be performed, and PO is converted to a single-timeslot optimization, which can then be decomposed into Problems RC-$i$ and Problem PA. The only difference is that the objective function of Problem PA is replaced by  $\max_{\mathbf{p}} \ \frac{1}{K}Q_1 \log\Big(1+\frac{p_1g_1}{\eta}\Big)+\frac{1}{K}Q_2 \log\Big(1+\frac{p_2g_2}{\eta}\Big) +\ldots+ \frac{1}{K}Q_K \log\Big(1+\frac{p_Kg_K}{\eta}\Big)-Z\left(\sum_{i=1}^{K}p_i\right)$. We can verify that Problem PA for OMA is a convex optimization problem, which can be solved by a standard method (e.g., interior point method). The derived optimal solution is adopted in each timeslot.

\subsubsection{Single User Serving (Single) Scheme} In this scheme,   only one single user is served in each timeslot. For fairness,  the Single scheme has the same queueing model as the NOMA-OPT scheme.  The objective function of Problem PA 
does not change, but an additional constraint that at most one of $p_1, p_2,\ldots, p_K$ can be positive is added. Through this way, the resultant power is allocated to one ``best'' user in each timeslot.\footnote{Note that ``Single'' is not optimal to Problem PA.
 If $g_i>g_j$, then serving user $i$ instead of user $j$ will lead to a better solution. If $Q_i>Q_j$, we should give a higher priority to serve user $i$ since the queue backlog at user $i$ is larger. If  $g_i>g_j$ and $Q_i>Q_j$, serving user $i$ but not $j$ is optimal. However, if we have  $g_i>g_j$ and $Q_i<Q_j$, it is not straightforward to decide which user to serve, and how much power should be allocated to users $i$ and $j$. We can only solve the problem optimally through our proposed NOMA-OPT solution, which is one of our core contributions of this work.}

\subsubsection{NOMA with Equal Power Allocation (NOMA-EQ) scheme} In this scheme, NOMA is employed, but
power allocation is not optimized. In each timeslot, power is equally allocated to all users, and the total power allocated is equal to $P_{\mathrm{mean}}$. For fairness,  the NOMA-EQ scheme has the same queueing model as the NOMA-OPT scheme.

\subsubsection{NOMA with Power Allocation Proportional to Queue Backlog (NOMA-Pro-Q) scheme} In this scheme, NOMA is employed, but
power allocation is not optimized. In each timeslot, the power allocated to each user is proportional to the queue backlog of that
user ($Q_i(t)$), and the total power allocated is equal to $P_{\mathrm{mean}}$.  The NOMA-EQ scheme has the same queueing model as the NOMA-OPT scheme.

\subsection{Tradeoff between Utility and Queue Backlog}\label{subsection_sim_tradeoff1}

\begin{figure} \centering
\addtolength{\subfigcapskip}{-0.1cm}
\subfigure[Data rate. ] { \label{fig_sim3a}
\includegraphics[width=0.45\textwidth]{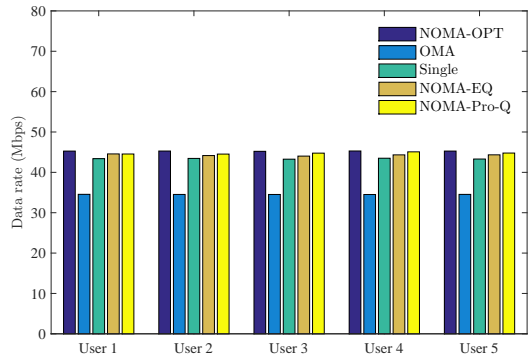}
}
\subfigure[Queueing delay.  ] { \label{fig_sim3b}
\includegraphics[width=0.45\textwidth]{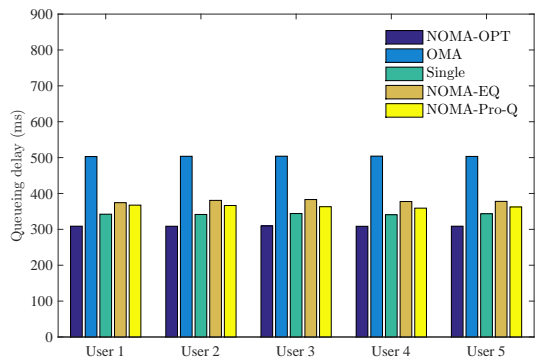}
}
\subfigure[Utility. ] { \label{fig_sim3c}
\includegraphics[width=0.45\textwidth]{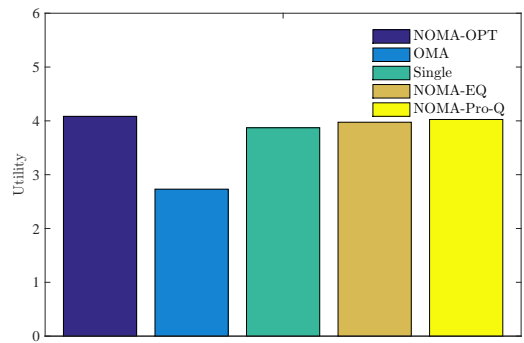}
}

\vspace{-0.2 cm}
\caption{Performance comparison in Scenario 1.}
\label{fig_sim3}
\end{figure}

\begin{figure} \centering
\addtolength{\subfigcapskip}{-0.1cm}
\subfigure[Data rate. ] { \label{fig_sim4a}
\includegraphics[width=0.45\textwidth]{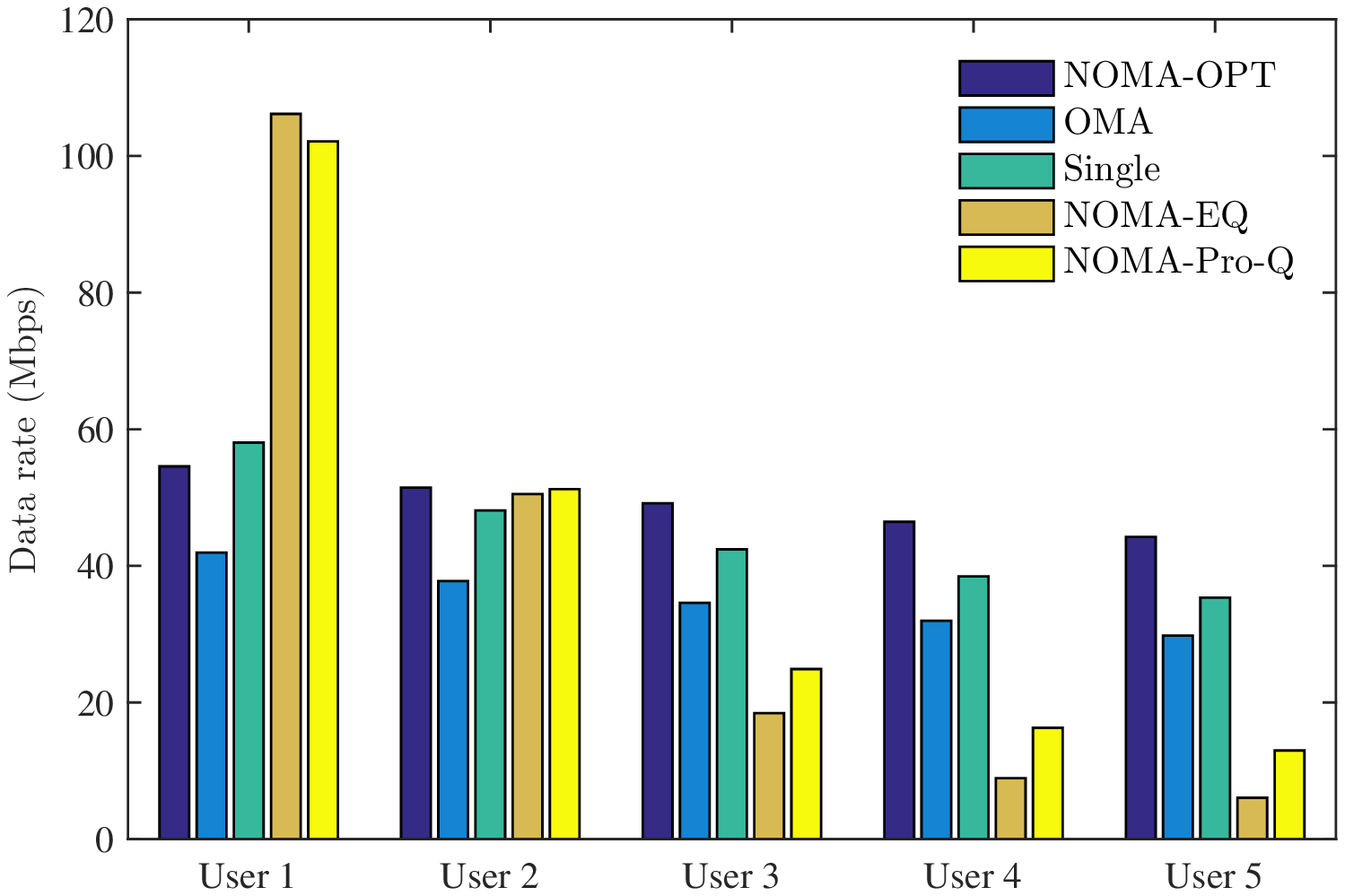}
}
\subfigure[Queueing delay. ] { \label{fig_sim4b}
\includegraphics[width=0.45\textwidth]{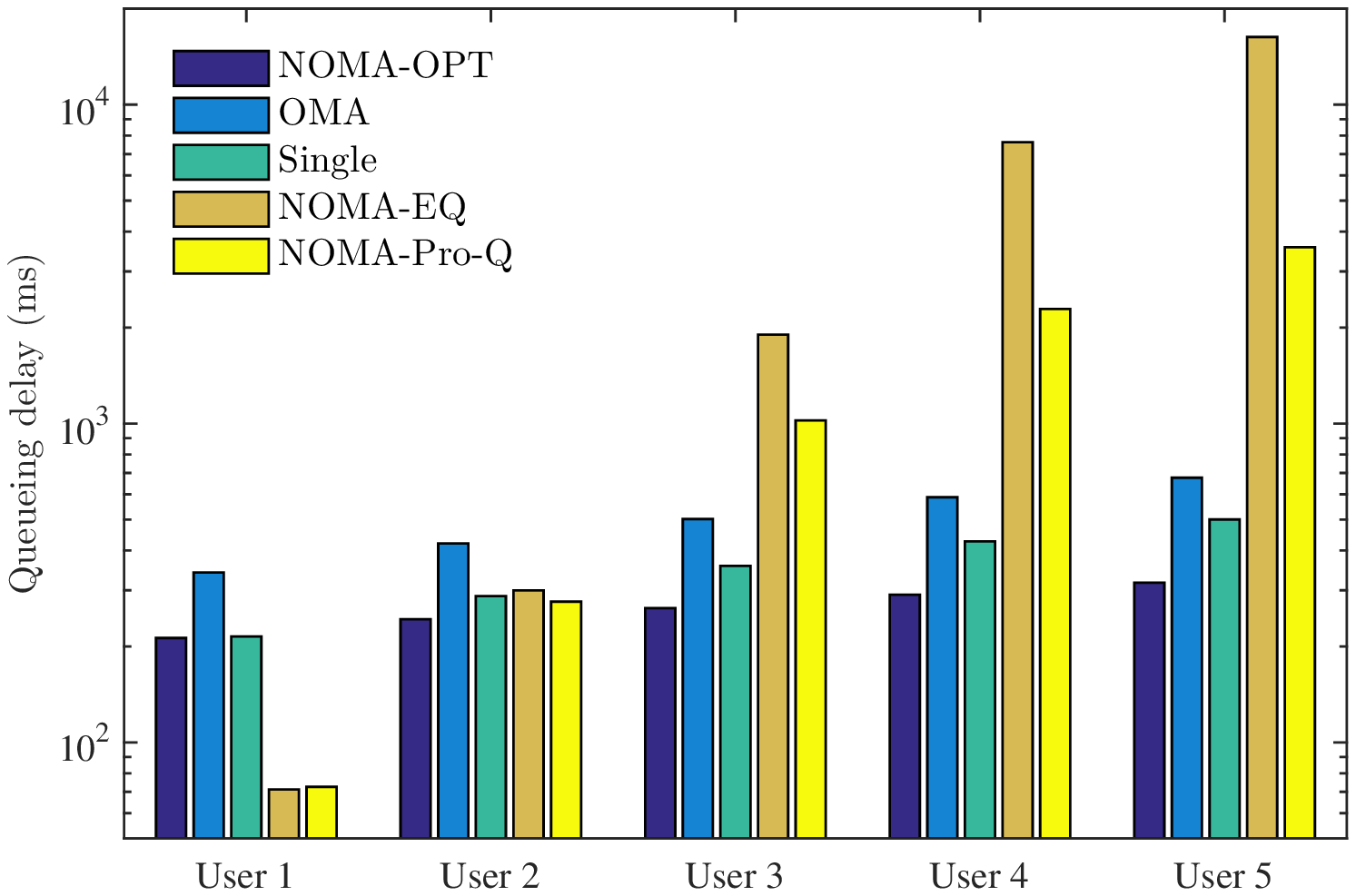}
}
\subfigure[Utility. ] { \label{fig_sim4c}
\includegraphics[width=0.45\textwidth]{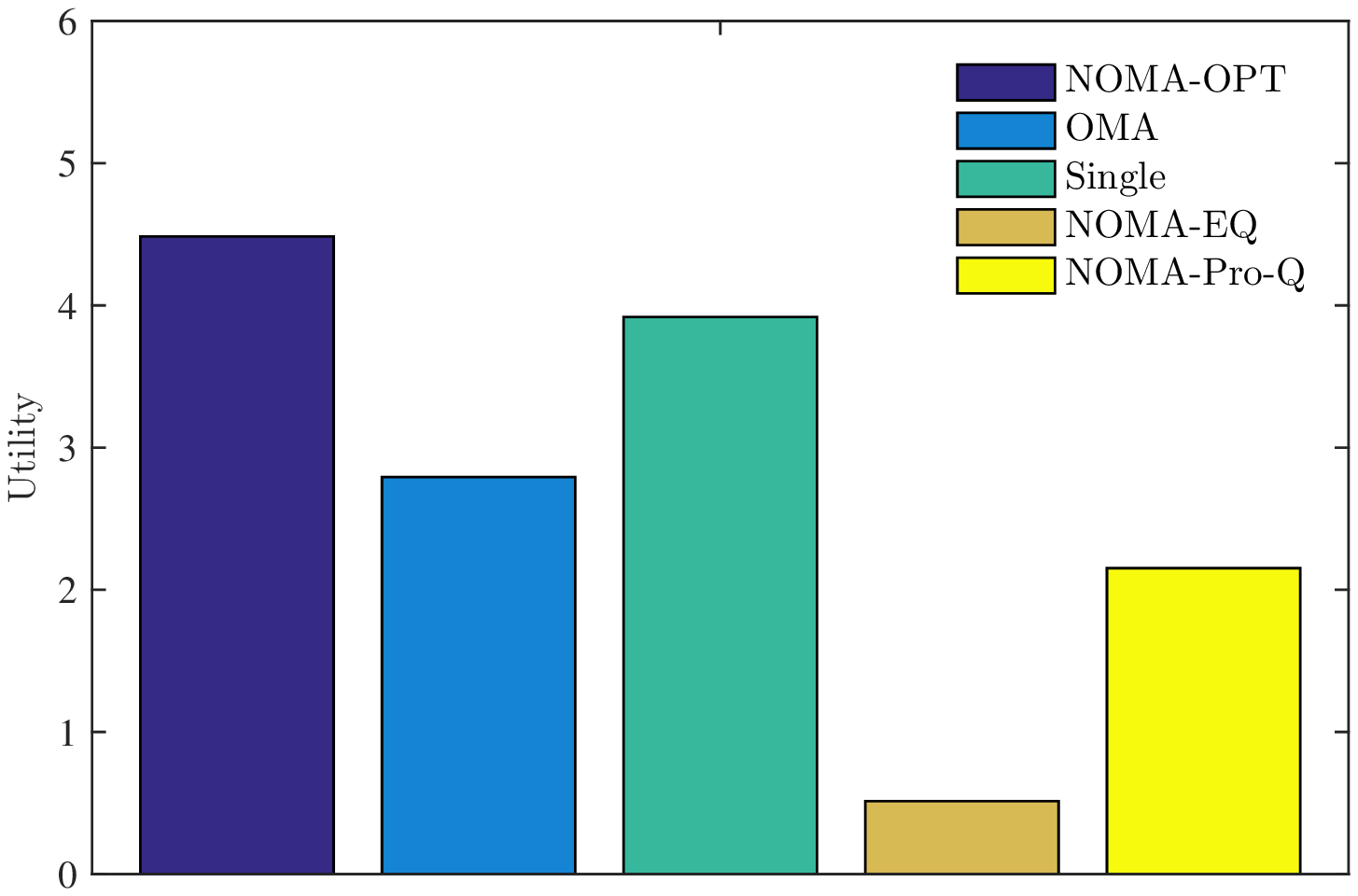}
}

\vspace{-0.2 cm}
\caption{Performance comparison in Scenario 2.}
\label{fig_sim4}
\end{figure}

First, we study the tradeoff between the long-term user utility and the average queue backlog $Q_i(t)$, when we adjust the value of $V$ from $10^{-1}$ to $10^3$. Recall that $V$ is a tunable value that achieves an $[O(V), O(\frac{1}{V})]$ tradeoff between the queue backlogs and utility, shown in (\ref{formula_tradeoff1})--(\ref{formula_tradeoff2}). We study two scenarios. There are $K=5$ users in both scenarios. In Scenario 1, all of the five users are $100$ meters from the BS; and in Scenario 2, the five users are $60$, $80$, $100$, $120$, and $140$ meters from the BS.

The tradeoff performance of NOMA-OPT and four benchmark schemes in Scenarios 1 and 2 is shown in Figs.~\ref{figure_sim1} and \ref{figure_sim2}.
First, we observe that when the queue backlog increases, the user utility increases fast at the beginning, but then gradually saturates. This matches our analysis that there is an $[O(V), O(\frac{1}{V})]$ tradeoff. Please note that here $O(V)$ indicates the size of queue backlog, and $O(\frac{1}{V})$ indicates the gap between the utility under $V$ and the maximum possible utility. Therefore, the utility asymptotically  approaches the maximum possible utility when we keep increasing the queue backlog. 

Second, in both figures, we observe that the performance of NOMA-OPT is much better than those of OMA and Single.
Given an arbitrary queue backlog, the overall utility of NOMA-OPT is higher than those of Single and OMA. This matches our expectation:
NOMA-OPT outperforms OMA because NOMA can better utilize the radio resource to achieve better data rates for all users. NOMA-OPT outperforms Single because Single gives a sub-optimal solution to Problem PA in each timeslot (since an additional constraint that at most one user can be served is added in Single).

Third, we observe that in Fig.~\ref{figure_sim1}, the performances of NOMA-OPT, NOMA-EQ, and NOMA-Pro-Q schemes are close to each other. This is because in Scenario 1, the distances from the five users to the BS are the same, so that the five users are symmetric. Equal power allocation is already a sufficiently good solution to these five symmetric users. In this scenario, for NOMA-Pro-Q scheme, the five users experience almost the same queue backlogs, so that NOMA-Pro-Q is almost equivalent to NOMA-EQ. However, in Fig.~\ref{figure_sim2} when the distances from the five users to the BS are different (which is most likely in reality), NOMA-OPT substantially outperforms NOMA-EQ and NOMA-Pro-Q. This is because NOMA-EQ does not allocate sufficient power to users far from the BS, so that these users experience very low data rate, leading to low log utility. Compared with NOMA-EQ, NOMA-Pro-Q  improves the log utility by allocating more power to users farther away. However, the improved performance is still much worse than that of NOMA-OPT.

Finally, we notice that without wisely devising the power allocation scheme for NOMA, the benefit of NOMA cannot be realized. The performances of NOMA-EQ and NOMA-Pro-Q are even much worse than that of OMA in Fig.~\ref{figure_sim2}. Such observation further demonstrates the significance of our proposed NOMA-OPT scheme.


\subsection{User Data Rate and Queueing Delay}
We also study the average user data rate and average queueing delay\footnote{It is the average time of a data bit staying in queue $Q_i(t)$. In the simulation, we track the timestamps at which each data bit enters and leaves the queue $Q_i(t)$, and compute the time difference. The result is averaged over all data bits in the simulation.} of the five users in Scenarios 1 and 2, under the five schemes. $V=30$ in this subsection. 
In Figs.~\ref{fig_sim3} and \ref{fig_sim4}, we show the data rate, average queueing delay, and the overall utility in subfigures (a), (b), and (c), respectively.
In Scenario 1, NOMA-OPT slightly improves the data rates and queueing delays for all users compared with Single, NOMA-EQ, and NOMA-Pro-Q, but substantially improves  data rates and queueing delays for all users compared with OMA. In Scenario 2, the advantage of NOMA-OPT is much more substantial. In particular, we observe
\begin{enumerate}
  \item Compared with OMA, the data rates and queueing delays for all users are substantially improved.
  \item Compared with Single, the data rate of user 1 is slightly smaller. This gives way to the performance gains of users 2, 3, 4, and 5, and thus the overall log utility shown in Fig.~\ref{fig_sim4c}. In addition, the delay performance of users 4 and 5 are substantially improved (Note that the $y$-axis of Fig.~\ref{fig_sim4b} is shown in log scale). This is because the performance gain of NOMA-OPT stems from the optimality in solving Problem PA (while Single gives a non-optimal solution).
  \item NOMA-EQ and NOMA-Pro-Q only provide large data rate and small delay to the closest user, while the data rate and delay performance of far users are very poor. Therefore, without our proposed NOMA-OPT approach, the performance NOMA could be unsatisfactory.
  \item  NOMA-OPT brings the highest performance enhancement to user 5 (the farthest user) compared with all the other four schemes. It suggests that edge users will get more benefits if NOMA-OPT is adopted.
\end{enumerate}


%

\subsection{Evolution of Queue Backlogs}
\begin{figure}[t]
\centering  \hspace{0pt}
\includegraphics[scale=0.53]{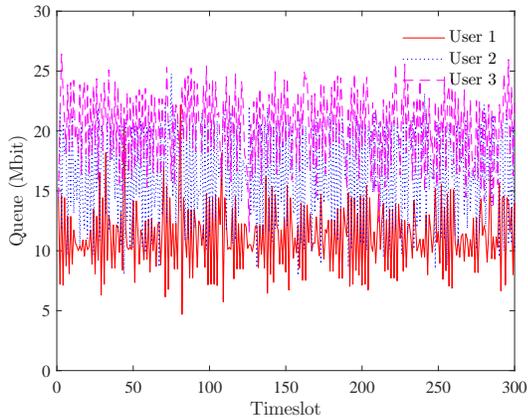}
\vspace{-0.2 cm}
\caption{Evolution of queue backlogs in Scenario 3. }
\label{figure_sim5}
\end{figure}

\begin{figure} \centering
\addtolength{\subfigcapskip}{-0.1cm}
\subfigure[Utility gain. ] { \label{fig_sim6a}
\includegraphics[width=0.45\textwidth]{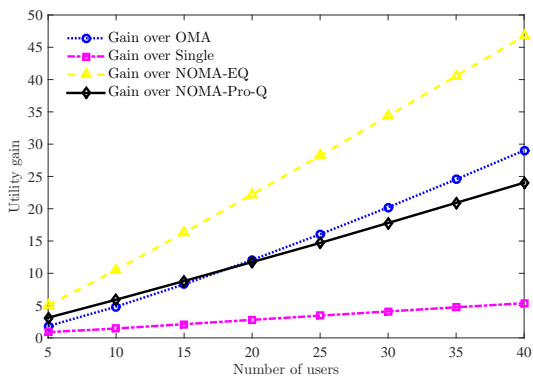}
}
\subfigure[Delay ratio. ] { \label{fig_sim6b}
\includegraphics[width=0.45\textwidth]{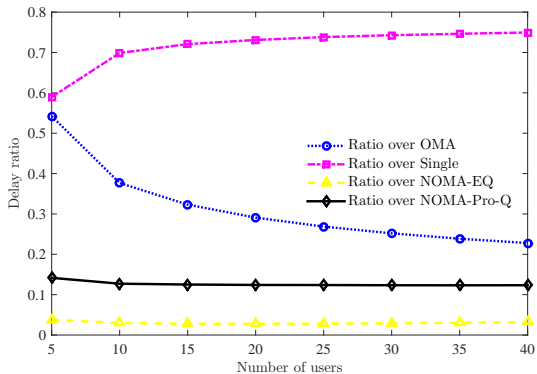}
}
\vspace{-0.2 cm}
\caption{Performance comparison under different numbers of users.}
\label{fig_sim6}
\end{figure}

In this subsection, we investigate the evolution of the queue backlogs and understand how they influence the system fairness. For illustration purpose, we study the following Scenario 3: There are $K=3$ users, and they are  $20$, $100$, and $200$ meters from the BS. Obviously, their distances to the BS and thus channel gains are quite different in this scenario. $V=50$ in this subsection.

Through the NOMA-OPT approach, the data rates of the three users are $215.65$,   $137.94$, and  $105.82$ Mbps respectively. Even if the channel gain of the furthest user is $40$ dB smaller than that of the closest user (on average), it can still achieve $48.8\%$ of data rate of the closest user, suggesting that the system fairness is well protected. This is aligned with the fact that the queue backlogs of the three users fluctuate around $11$,   $16$, and   $20$ respectively: user 3 is usually associated with a larger queue backlog and thus has a higher weight (priority) in the optimization problem in (\ref{formula_PA_obj}).

\subsection{Performance Gain under Different Numbers of Users}

Finally, we investigate the performance gain of NOMA-OPT over the four benchmark schemes under variable numbers of users. $V=20$ in this subsection. We consider scenarios where the number of users increases from $5$ to $40$. In each case, users are lined from $50$ meters to $150$ meters from the
BS with equal intervals. We study the utility gain of NOMA-OPT compared with each of the benchmark schemes (i.e., utility of NOMA-OPT minus utility of the benchmark scheme) and the delay ratio of NOMA-OPT over each of the benchmark schemes (i.e., average queueing delay of NOMA-OPT over average queueing delay of the benchmark scheme).
   The results shown in  Fig.~\ref{fig_sim6} suggest that NOMA-OPT brings substantial performance gain, in terms of both data rate and delay,  under a variety of user numbers in the system. In addition, the utility gain increases as the number of user increases. This is because the original Problem PO is solved asymptotically optimally by our proposed NOMA-OPT, but all of the benchmark schemes give non-optimal solutions. The more users in the system, a non-optimal solution has a higher chance to make a ``poor'' decision, and thus causes larger performance loss compared with NOMA-OPT.

\section{Conclusions}
In this paper, we developed a long-term resource allocation framework for a downlink multi-user non-orthogonal multiple access (NOMA) network, which jointly optimizes the rate control at the network layer as well as the power allocation at the physical layer to maximize the long-term network utility, subject to several practical long-term and short-term constraints on the power consumption and user queue stability. By resorting to the Lyapunov optimization theory, we attained the asymptotically optimal solution to the formulated long-term network utility maximization problem. To achieve this, we converted the utility maximization problem into a series of short-term rate control and power allocation problems to be optimized in each timeslot and successfully resolved these short-term problems by leveraging the special structures of the objective functions. Simulation results were presented to compare the performance of the proposed NOMA method with those of four benchmark schemes, including OMA and non-optimal NOMA schemes, under the same setups. The simulation results showed that compared to these benchmark schemes, the proposed NOMA method can achieve a higher network throughput, lower data delay, and better user fairness.
\ifCLASSOPTIONcaptionsoff
  \newpage
\fi
\bibliography{IEEEabrv,ref}
\bibliographystyle{IEEEtran}
\end{document}